\documentclass[11pt,a4paper]{article}

\usepackage[T1]{fontenc}
\usepackage{lmodern}
\usepackage[margin=1.65cm, top=1.75cm, bottom=1.85cm,
            headheight=13pt, headsep=12pt, footskip=16pt]{geometry}
\usepackage{amsmath, amssymb, amsthm}
\usepackage{graphicx}
\usepackage{booktabs}
\usepackage{array}
\usepackage{caption}
\usepackage{subcaption}
\usepackage[hidelinks, colorlinks=true, linkcolor=blue!60!black,
            citecolor=blue!60!black, urlcolor=blue!60!black]{hyperref}
\usepackage[round, authoryear]{natbib}
\usepackage{xcolor}
\usepackage{soul}          
\sethlcolor{yellow}
\usepackage{float}
\usepackage{multirow}
\usepackage{fancyhdr}
\usepackage{titlesec}
\usepackage{titling}
\usepackage{bm}
\IfFileExists{bbm.sty}{\usepackage{bbm}}{}
\providecommand{\mathbbm}[1]{\mathbb{#1}}
\newcounter{ALCline}
\IfFileExists{algorithm.sty}{%
  \usepackage{algorithm}\usepackage{algpseudocode}%
}{%
  \newfloat{algorithm}{htbp}{loa}\floatname{algorithm}{Algorithm}%
  \newenvironment{algorithmic}[1][0]{%
    \begin{list}{\footnotesize\arabic{ALCline}:}{\usecounter{ALCline}%
      \setlength{\leftmargin}{2.2em}\setlength{\itemsep}{1pt}\setlength{\topsep}{2pt}}}%
    {\end{list}}%
  \newcommand{\State}{\item}%
  \newcommand{\For}[1]{\item \textbf{for}\ ##1\ \textbf{do}}%
  \newcommand{\EndFor}{\item \textbf{end for}}%
  \newcommand{\Return}{\textbf{return}\ }%
  \newcommand{\Comment}[1]{\hfill{\footnotesize$\triangleright$\ ##1}}%
}

\makeatletter
\newcommand{\blfootnote}[1]{%
  \begingroup
  \renewcommand{\thefootnote}{}%
  \renewcommand{\@makefnmark}{}%
  \footnotetext{#1}%
  \endgroup
}
\makeatother

\pretitle{\begin{center}\fontsize{17}{20}\selectfont\bfseries}
\posttitle{\par\end{center}\vspace{0.4em}}
\preauthor{\begin{center}}
\postauthor{\par\end{center}\vspace{0.4em}}
\predate{}
\postdate{}

\titleformat*{\section}{\large\bfseries}
\titleformat*{\subsection}{\normalsize\bfseries}
\titleformat*{\subsubsection}{\normalsize\itshape}
\titlespacing*{\section}{0pt}{9pt}{4pt}
\titlespacing*{\subsection}{0pt}{7pt}{3pt}
\titlespacing*{\subsubsection}{0pt}{5pt}{2pt}

\fancypagestyle{plain}{%
  \fancyhf{}%
  \fancyhead[R]{\small\thepage}%
  \fancyhead[L]{\small\itshape Regulus Metrics: A T20 Batting and Bowling Impact Framework}%
}

\newtheoremstyle{finding}%
  {6pt}{6pt}{\itshape}{0pt}{\bfseries}{.}{0.5em}{}
\theoremstyle{finding}
\newtheorem{finding}{Finding}

\newtheoremstyle{propstyle}%
  {5pt}{5pt}{\itshape}{0pt}{\bfseries}{.}{0.5em}{}
\theoremstyle{propstyle}
\newtheorem{proposition}{Proposition}
\newtheorem{remark}{Remark}

\newcommand{\E}{\mathbb{E}}
\newcommand{\Var}{\operatorname{Var}}

\newcommand{\sd}{\operatorname{sd}}

\newcommand{\RAE}{\mathrm{RAE}}                    
\newcommand{\val}{v}                               
\newcommand{\realDAR}{\mathrm{realDAR}}            
\newcommand{\xDAR}{\mathrm{xDAR}}                  
\newcommand{\DAR}{\mathrm{DAR}}                    
\newcommand{\wc}{c}                                
\newcommand{\Vfun}{V}                              
\newcommand{\OmegaBat}{\Omega^{(\mathrm{Bat})}}
\newcommand{\OmegaBowl}{\Omega^{(\mathrm{Bowl})}}

\title{Context-adjusted Player Evaluation for Twenty20 Cricket}
\author{%
  \large Rhitankar Bandyopadhyay\\[0.35em]
  \normalsize University of Florida, Gainesville, FL~32611, USA.
  Email: \href{mailto:r.bandyopadhyay@ufl.edu}{r.bandyopadhyay@ufl.edu}%
}
\date{}

\begin{document}
\maketitle

\blfootnote{This work is part of the development of an app \textsc{Regulus}. The \textsc{Regulus}
platform and its constituent metrics are proprietary at this point; all methodology described
herein is for documentation and scholarly communication only.}

\thispagestyle{empty}

\begin{abstract}
\noindent
I develop a reproducible framework for evaluating individual batting and bowling performances in Twenty20 (T20) cricket on one interpretable scale of runs above expectation, built from two ball-level primitives. The first, \emph{Runs Above Expected} ($\RAE$), is the residual between the runs scored on a delivery and a contextual expectation of what an average performer would produce in the same situation. That expectation is a multiplicative log-linear model of the cohort scoring rate, whose per-cell estimator is shown to be a conditional Poisson maximum likelihood multiplier. It is fitted by iterated backfitting and stabilised by empirical Bayes shrinkage, so that thinly sampled contexts are pooled towards the population. A single opposition symmetry places run-scoring and run-prevention on the same footing. The second primitive, \emph{Dismissal Adjusted Runs} ($\DAR$), prices a dismissal in that currency as the runs it forgoes, read off a batting side value function solved by dynamic programming through a Bellman expectation recursion, under observed play. Because dismissal is the expected end of every innings, the realised wicket cost ($\realDAR$) is centered against its expectation ($\xDAR$) under a league dismissal hazard rate. The centered quantity is a run-weighted mean zero martingale residual of the dismissal process, so a player is charged only for departing from average behaviour. The two primitives sum to a symmetric Impact, which makes batting and bowling comparable in centre as well as in unit. Estimated on over 2.7 million legal deliveries of men's T20 cricket, the framework recovers known contextual structure, agrees with the conventional rates it refines while correcting their context-blindness, and yields face-valid player, innings and season leaderboards for the Indian Premier League.
\end{abstract}

\vspace{2pt}

\vskip 18mm

\section{Introduction}
\label{sec:intro}

Twenty20 (T20) is the shortest and most widely watched form of professional cricket. Each side bats
for a single innings of twenty overs, i.e.\ at most $120$ legal deliveries, so scoring is compressed and
unforgiving. A handful of balls can decide a match, and every delivery carries a visible trade-off
between runs and the risk of dismissal. Since its introduction, the format has grown into a global
professional circuit of franchise leagues, the Indian Premier League (IPL) foremost among them, alongside the Big Bash League (BBL), the Caribbean Premier League (CPL), the SA20 to name a few, in
which recruitment, selection and in-game tactics turn on large sums of money and on increasingly
fine grained data. Every ball of professional T20 matches is now recorded, which makes the quantitative
evaluation of individual batting and bowling both possible and consequential, and it is that
evaluation this paper develops.

The rest of this section sets the context:
Section~\ref{sec:intro:analysis} discusses performance analysis in T20 and why raw statistics are
inadequate for it, Section~\ref{sec:intro:lit} reviews the relevant literature across cricket and
other sports, and Section~\ref{sec:intro:contrib} states our contribution. The data are described in
Section~\ref{sec:data}; Section~\ref{sec:rae} formulates the first primitive, Runs Above Expected;
Section~\ref{sec:dar} constructs the second of them, Dismissal Adjusted Runs, through a batting-side value function;
Section~\ref{sec:composite} combines them into a single Impact rating;
Section~\ref{sec:validation} reports the empirical properties and validation; and
Section~\ref{sec:limits} concludes.

\subsection{Performance analysis in T20 cricket}
\label{sec:intro:analysis}

A T20 innings unfolds in three broadly recognised phases. Unless rain or some other interruption shortens the match, the `powerplay' of the first six
overs restricts the fielding side and rewards aggression; the `middle overs' usually see
consolidation, often against spin, and the `death overs' at the end of the innings are played
at maximum risk. Batting and bowling therefore mean very different things at different points of an
innings, and against different opponents and conditions, and this is what makes individual evaluation
difficult. The conventional measures, namely a batter's average (runs per dismissal) and strike rate (runs
per hundred balls), a bowler's average (runs per dismissal), economy rate (runs per $6$ balls), and strike rate (balls per dismissal), aggregate over this
heterogeneity and are systematically confounded by it. A strike rate of $150$ is unremarkable at the
death and exceptional in the middle overs, an economy around $8$ is expensive in the middle overs and
cheap at the death, a boundary off a good bowler is worth more than the same boundary off a
part-timer, and a dismissal early in an innings costs a side far more than the same dismissal towards its
end. None of these situational factors are visible in the raw numbers. Two further problems compound it. Batting and
bowling are reported on incommensurable scales, i.e.\ runs per ball (or $100$ balls) against runs per over, so the two
roles cannot be placed on a common footing, which matters acutely for all-rounders and for
squad-level decisions, and a raw average credits occupation of the crease without asking what the
runs were worth in context. Yet these are the distinctions performance analysis exists to draw, since selection, auction valuation, match-ups and in-game strategy all turn on what a player
actually contributes, net of the situation he found himself in. A metric fit to inform such decisions
must adjust for context at the level of the individual ball, must reward run-scoring and
run-prevention on the same scale as the wickets that punctuate them, and must stay interpretable
enough to be trusted and reproducible enough to be checked.

\subsection{Review of literature}
\label{sec:intro:lit}

Quantitative cricket analysis has long rested on the idea of a run expectancy, a function that
assigns to each state of an innings the runs a side can expect from it. Its best known instance is
the Duckworth-Lewis method for resetting targets in interrupted matches \citep{duckworth1998}, whose
resource table is exactly such a function of the overs and wickets a side has in hand. A parallel
line used dynamic programming to derive optimal play: \citet{clarke1988} computed optimal
scoring rates in one-day cricket, and \citet{preston2000} formulated optimal batting strategy against a target.
The same machinery underlies the modern win-and-score predictor WASP, built by \citet{brooker2011} on
a dynamic programming run expectancy estimated from historical limited-overs data, while later
statistical models describe how scoring and dismissal risk evolve within an innings rather than
optimise them \citep{stevenson2021}. The value function of the present paper belongs to this
tradition, but is put to a different use. I solve it under observed play, to price the runs a
wicket forgoes, not to prescribe how a side should bat.

Turning from game states to players, early evaluation measures refined the batting average and strike
rate into single indices \citep{lemmer2004}, but these inherit the context-blindness of their
ingredients. The closest antecedent of our work is the simulation approach of
\cite{davis2015player} and \cite{davis2015simulator}, who built a ball-by-ball T20 simulator and, from it, an \emph{expected run differential}: the runs a player adds relative to a standard player in the same role, obtained by re-simulating innings with and without him. This captures context and reports contributions in runs, but it depends on a full generative simulator and a role-defined replacement, and it does not yield a transparent per-ball decomposition. Other work rates players through Bayesian or machine-learning models and data-envelopment indices, or models player valuation specifically for the IPL auction. These gain flexibility at the cost of interpretability, and typically treat batting and bowling on separate scales.

Alongside the academic literature, several proprietary systems target the same quantity we
do: a context-adjusted, unified measure of impact, but behind closed models. Cricinfo's Smart
Stats\footnote{Read articles and illustrations on Cricinfo's Smart Stats here: \href{https://www.cricinfo.com/genre/superstats-706}{https://www.cricinfo.com/genre/superstats-706}} report \emph{Smart Runs} (runs re-weighted by match context) and
\emph{Smart Wickets} (a dismissal valued by the batter removed, its timing, and the state of the
game) and combine the two into a single `Player Impact' on a runs scale,
the same design goal as ours. CricViz derives expected runs, expected wickets and expected dismissals
from ball-tracking data, and independent analysts have proposed replacement-level
measures such as Runs Above Average Replacement\footnote{Read here: \href{https://www.cricinfo.com/story/using-the-runs-above-average-replacement-metric-to-assess-the-quality-of-test-batsmen-1226970}{https://www.cricinfo.com/story/using-the-runs-above-average-replacement-metric-1226970}}, a cricketing analogue of baseball's value over
replacement. These efforts are influential, but their mechanisms are largely
unpublished, they often rely on proprietary or tracking data unavailable for most matches, and they
are validated only informally, so they can be neither reproduced nor audited by others.

The underlying idea of valuing an event by its expected contribution to a scoring currency, above what
an average or replacement performer would have produced, is mature in several other sports. Baseball's
run expectancy and win probability accounting descends from \citet{lindsey1963} and now underpins
reproducible wins-above-replacement systems \citep{baumer2015}. In association football, expected
goals grade a chance by its scoring probability, and action value frameworks such as VAEP value every
on-the-ball action by the change it induces in the probabilities of scoring and conceding
\citep{decroos2019}. In basketball, the expected possession value model prices the state of a
possession by the points it is expected to yield \citep{cervone2016}. All of these share one logic: price
events in the currency of the game, relative to expectation, which the present framework brings to
T20 cricket.

The cricketing literature either optimises play rather than evaluating it, evaluates
players through simulators or opaque models that are hard to reproduce, or reports the two roles on
scales that do not meet, and the systems that come closest to a unified impact measure are
proprietary. What is missing is a transparent, likelihood based framework that adjusts for context
at the level of the individual ball, prices wickets in the same runs currency as the runs themselves,
and places batting and bowling on one comparable scale, reproducible in full from public data. This
paper provides it.

\subsection{Contribution of the work}
\label{sec:intro:contrib}

I build the framework from two ball-level primitives, both expressed in runs above expectation and
shared symmetrically between the two roles. The first, \emph{Runs Above Expected} ($\RAE$), is the
residual between the runs scored on a delivery and a contextual expectation of what an average
performer would have produced in the same situation, specified through phase, innings, era, fall of wickets, bowler type, venue and the specific opponent. That expectation is a multiplicative, log-linear model of the cohort scoring rate whose per-cell estimator I show to be a conditional Poisson maximum likelihood multiplier, fitted by iterated backfitting and stabilised by empirical Bayes shrinkage, so that
thinly sampled contexts and players are pooled towards the population rather than trusted blindly. A
single opposition symmetry credits a batter against the bowler he faced and a bowler against the
batter he contained, which places run-scoring and run-prevention on one scale. The second primitive,
\emph{Dismissal Adjusted Runs} ($\DAR$), prices a dismissal in that same currency. A batting-side value function,
solved by dynamic programming under observed play, gives the runs a wicket forgoes in a given state.
But charging a batter the full cost of every dismissal would penalise even an average player, since
being dismissed is the expected end of an innings. I therefore centre the realised wicket cost
against its expectation under a corresponding league dismissal hazard rate, which yields a quantity that is exactly a run-weighted martingale residual of the dismissal process, mean zero under average behaviour, so a player is charged only for departing from it. Summed over a career, the two primitives give a symmetric Impact, batting and bowling differing only in sign, on a single interpretable runs scale.

Three parts of this are new. The first is reproducibility. Every multiplier, hazard and shrinkage
strength is estimated from public ball-by-ball data, and there is no tuning constant left to set,
whereas the proprietary impact systems do not disclose their models at all. The second is the
pricing of a wicket as a centered martingale residual, which is to our knowledge new to cricket.
It is what puts a batter's and a bowler's Impact on a common centre and not only a common unit:
an average performer of either role sits near zero, and the long-standing difficulty of placing
the two roles on one footing goes away with it. The third is that both primitives are additive at
the ball level, so the same construction reads coherently from a single delivery to a long career
without re-fitting anything. I fit the framework on more than $2.7$ million legal T20 deliveries
and report the contextual structure it recovers, the agreement of each primitive with the
conventional rate it refines, and leaderboards and case studies for the IPL.

\section{Data}
\label{sec:data}

The analysis uses a ball-by-ball dataset across men's T20 cricket spanning from February
$2005$ to June $2026$ (data last updated as of June $05$, $2026$), drawn from $19$ competitions including T20 Internationals and several major domestic franchise and national leagues. Table~\ref{tab:dataset} summarises the dataset. The unit of analysis is the legal delivery (\texttt{is\_legal\_ball} $=$ True, wides and no-balls excluded).

\begin{table}[!ht]
\centering
\caption{Data summary and coverage. The dataset has two tiers: a \emph{baseline} record present for every delivery and a \emph{ball-tracking} record available only for a section of matches. Qualification thresholds count players by legal balls, $n_i$.}
\label{tab:dataset}
\small
\begin{tabular}{lr}
\toprule
Quantity & Value \\
\midrule
\multicolumn{2}{l}{\emph{Entire dataset}}\\
Total deliveries                              & 2{,}819{,}700 \\
Legal deliveries                              & 2{,}727{,}420 \\
Matches                                       & 12{,}492 \\
Distinct competitions                         & 19 \\
Date range                                    & Feb 2005 - May 2026 \\
IPL deliveries\,/\,matches                    & 295{,}629\,/\,1{,}243 \\
\midrule
\multicolumn{2}{l}{\emph{Coverage of the detailed (ball-tracking) tier}}\\
Matches with any stroke annotation            & 2{,}128 \quad(17.0\%\ of matches) \\
Deliveries with stroke annotation             & 448{,}751 \quad(16.5\%\ of legal) \\
Deliveries with line/length                   & 924{,}134 \quad(33.9\%\ of legal) \\
Baseline-only (Cricsheet) matches             & 10{,}364 \quad(83.0\%\ of matches) \\
\midrule
\multicolumn{2}{l}{\emph{Qualifying players}}\\
Qualifying batters, $n_i \geq 30/100/500$     & 6{,}234\,/\,3{,}800\,/\,1{,}220 \\
Qualifying bowlers, $n_i \geq 60/200$         & 4{,}607\,/\,2{,}516 \\
\bottomrule
\end{tabular}
\end{table}

\paragraph{Two tiers of detail.}
The dataset is deliberately heterogeneous in granularity, so I state exactly what is available where. Each of the $12{,}492$ matches
($2{,}727{,}420$ legal deliveries) carries a baseline record derived from public Cricsheet repositories: the batter and bowler identities, the key components of model~\eqref{eq:rae_model} (\texttt{phase}, \texttt{innings}, \texttt{date},
\texttt{wickets\_fallen}, \texttt{bowler\_type\_grp}, \texttt{ground}), the outcome
fields \texttt{runs\_off\_bat}, \texttt{is\_boundary}, \texttt{is\_dot}, bowler-credited dismissal flag and the chase-state components (runs, balls, wickets remaining, target at any point during 2nd innings) used by the second-innings value function of Section~\ref{sec:dar}. A smaller \emph{ball-tracking}
tier adds the delivery's line and length, stroke type, stroke magnitude, stroke direction, elevation etc. This richer information is present for only a minority of the data: stroke annotation covers $448{,}751$ legal deliveries ($16.5\%$) across $2{,}128$ matches ($17.0\%$)
and line and length data covers $924{,}134$ deliveries ($33.9\%$), so the majority of matches ($10{,}364$ or $83\%$) carry the baseline record alone. Coverage is also
uneven within the tier: some matches are annotated ball for ball while others are only partially tracked, which is why I report ball counts rather than treating any match as wholly tracked.

This heterogeneity forces a design choice. Both primitives, $\RAE$ (Section~\ref{sec:rae}) and Dismissal Adjusted Runs (Section~\ref{sec:dar}), use only baseline fields, so they are computed identically on every match irrespective of ball-tracking availability. Restricting the covariates to those recorded for every delivery keeps the sample pool as large as possible, lets a player be compared across eras, venues and match situations on one consistent baseline, and keeps the framework reproducible from public data. The richer ball-tracking covariates (line, length, stroke geometry etc.) are therefore not used by the metrics defined here. A fairer batting evaluation would, however, credit run-scoring relative to the exact delivery faced. I accordingly extend the expected runs model to incorporate the delivery's line and length on the annotated subset, and formulate a modified $\RAE$ in Section~\ref{sec:supp:mrae} of the supplementary materials. No metric silently degrades on a baseline-only match; each is fully defined from fields recorded everywhere.

\section{The `Runs Above Expected' (\texorpdfstring{$\RAE$}{RAE}) model}
\label{sec:rae}

$\RAE$ is the foundation of the framework, and every other metric is a functional of the ball-level $\RAE$ series. The construction has three parts: a contextual expectation model (Section~\ref{sec:rae:model}), its estimation by iterated, shrunk backfitting (Section~\ref{sec:rae:fit}-\ref{sec:rae:shrink}) and the role-specific opposition overlay followed by aggregation (Section~\ref{sec:rae:agg}).

\subsection{The contextual expectation model}
\label{sec:rae:model}

Let $r_b\in\{0,1,2,3,4,6\}$ be the runs off the bat on legal delivery $b$ and let $\bm{x}_b$ denote its context. I model the cohort expected runs, defined as what an average batter would score (for batting), and an average bowler would concede (for bowling), in that situation, as a product of a global base rate and contextual multipliers:
\begin{equation}
\label{eq:rae_model}
\mu(\bm{x}_b)=\mu_0\,
\gamma^{\mathrm{scen}}_{c_1(b)}\,
\gamma^{\mathrm{era}}_{c_2(b)}\,
\gamma^{\mathrm{wkt}}_{c_3(b)}\,
\gamma^{\mathrm{type}}_{c_4(b)}\,
\gamma^{\mathrm{ven}}_{c_5(b)}\,
\gamma^{\mathrm{opp}}_{c_6(b)} ,
\end{equation}
where $\mu_0$ is the grand mean over all legal deliveries, each of $\gamma^{\mathrm{scen}}_{c_1(b)}$, $\gamma^{\mathrm{era}}_{c_2(b)}$, $\gamma^{\mathrm{wkt}}_{c_3(b)}$, $\gamma^{\mathrm{type}}_{c_4(b)}$, $\gamma^{\mathrm{ven}}_{c_5(b)}$ is a non-negative multiplier attached to a categorical cell of the corresponding factor, and $\gamma^{\mathrm{opp}}_{c_6(b)}$ is the opposition factor, defined as the opposition bowler's multiplier for batter evaluation and the opposition batter's multiplier for bowler evaluation
(Section~\ref{sec:rae:agg}). The five shared factors stated in (\ref{eq:rae_model}) are:
\begin{center}\small
\begin{tabular}{lll}
\toprule
Factor & Cell key $c_f(b)$ & Captures \\
\midrule
Scenario   & (phase, innings)               & situational scoring level \\
Era        & (year, phase)         & temporal scoring inflation \\
Wicket     & (phase, wickets fallen)        & fall-of-wickets pressure \\
Bowler type& (phase, pace/spin sub-type)    & match-up difficulty by phase \\
Venue      & (ground, 3-year bucket)        & pitch, via temporal locality \\
\bottomrule
\end{tabular}
\end{center}
The wicket factor is conditioned on the exact number of wickets fallen ($0$ to $10$), so the fall-of-wickets gradient is estimated at full resolution and the empirical Bayes shrinkage of Section~\ref{sec:rae:shrink} pools the sparse deep-collapse cells automatically; the venue bucket groups consecutive 3-year windows at each ground. Taking logs on both sides,
\eqref{eq:rae_model} is a main-effects log-linear model,
\begin{equation}
\label{eq:rae_loglin}
\log\mu(\bm{x}_b)=\log\mu_0+\sum_{f}\log\gamma^{(f)}_{c_f(b)},
\end{equation}
on crossed categorical keys; the crosses (e.g.\ phase$\times$innings,
year$\times$phase etc.) admit a controlled set of two-way interactions while keeping the
parameter count tractable.

\begin{remark}[Difficulty parameter lives in the expectation]
\label{rem:nophaseweight}
Phase difficulty is captured \emph{inside} $\mu(\bm{x}_b)$, principally by the
scenario factor $\gamma^{\mathrm{scen}}_{(\text{phase},\text{innings})}$ (and refined by era, wicket and bowler-type factors, each phase-conditioned). A run in a higher-scoring phase is therefore measured against a higher bar automatically. No separate phase weight is applied to either role, $\RAE$ therefore acts as a pure residual on the run scale.
\end{remark}

The factors of \eqref{eq:rae_model} are not chosen by convention alone; each of them carries strong,
independent signal. Table~\ref{tab:anova} reports an analysis of deviance for the fitted
expected runs model: dropping any single factor from the converged fit raises the Poisson deviance
of Proposition~\ref{prop:poisson} by the stated amount, on its degrees of freedom. At this sample
size, every factor is significant far beyond any conventional threshold ($p<10^{-300}$), so the
deviance each explains is read not as a test but as a measure of relative importance. The venue and
the fall of wickets state carry the most, the phase-innings scenario and opposition bowler identity
follow, and even the smallest, the pace or spin sub-type of a bowler, remains decisive. The symmetric
opposition batter factor of the bowler evaluation fit is significant to the same degree.

\begin{table}[!ht]
\centering
\caption{Analysis of deviance for the expected runs model. Each row gives the rise in Poisson
deviance ($\chi^2$) when the factor is dropped from the converged fit, on its degrees of freedom;
all factors are significant beyond any usual threshold.}
\label{tab:anova}
\small
\begin{tabular}{lrrr}
\toprule
Factor & d.f. & $\Delta$Deviance ($\chi^2$) & $p$ \\
\midrule
\multicolumn{4}{l}{\emph{Always-available model~\eqref{eq:rae_model} \qquad \qquad ($n=2{,}727{,}420$ legal deliveries)}}\\
\midrule
Scenario\ (phase $\times$ innings)         & \ \ \ 7 & 21{,}686 & $<10^{-300}$ \\
Era\ (year $\times$ phase)                 & \ \ 68 & \ 8{,}611 & $<10^{-300}$ \\
Wicket\ (phase $\times$ wickets fallen)    & \ \ 32 & 27{,}972 & $<10^{-300}$ \\
Bowler type\ (phase $\times$ sub-type)     & \ \ 20 & \ 2{,}812 & $<10^{-300}$ \\
Venue\ (ground $\times$ 3-yr bucket)       & 1{,}267 & 41{,}767 & $<10^{-300}$ \\
Opposition bowler\ (identity)              & 7{,}281 & 27{,}545 & $<10^{-300}$ \\
\midrule
\multicolumn{4}{l}{\emph{Additional ball-tracking covariates\quad(annotated subset, $n=804{,}308$)}}\\
\midrule
Length\ (phase $\times$ length)            & \ \ 20 & 41{,}924 & $<10^{-300}$ \\
Line\ (phase $\times$ line)                & \ \ 20 & \ 3{,}409 & $<10^{-300}$ \\
\bottomrule
\end{tabular}
\end{table}

Beyond the covariates in ~\eqref{eq:rae_model}, the delivery's line and length are the most informative
additions the data offer. On the annotated subset, length alone explains a deviance rivalling that
of the venue factor and line a smaller but still overwhelming amount (lower block of
Table~\ref{tab:anova}). I keep them out of the headline model for the comparability reasons of
Section~\ref{sec:data}, but develop the resulting extension, a modified $\RAE$, and quantify its effect in
Section~\ref{sec:supp:mrae} of the Supplementary Materials.

\begin{figure}[!ht]
\centering
\includegraphics[width=\linewidth]{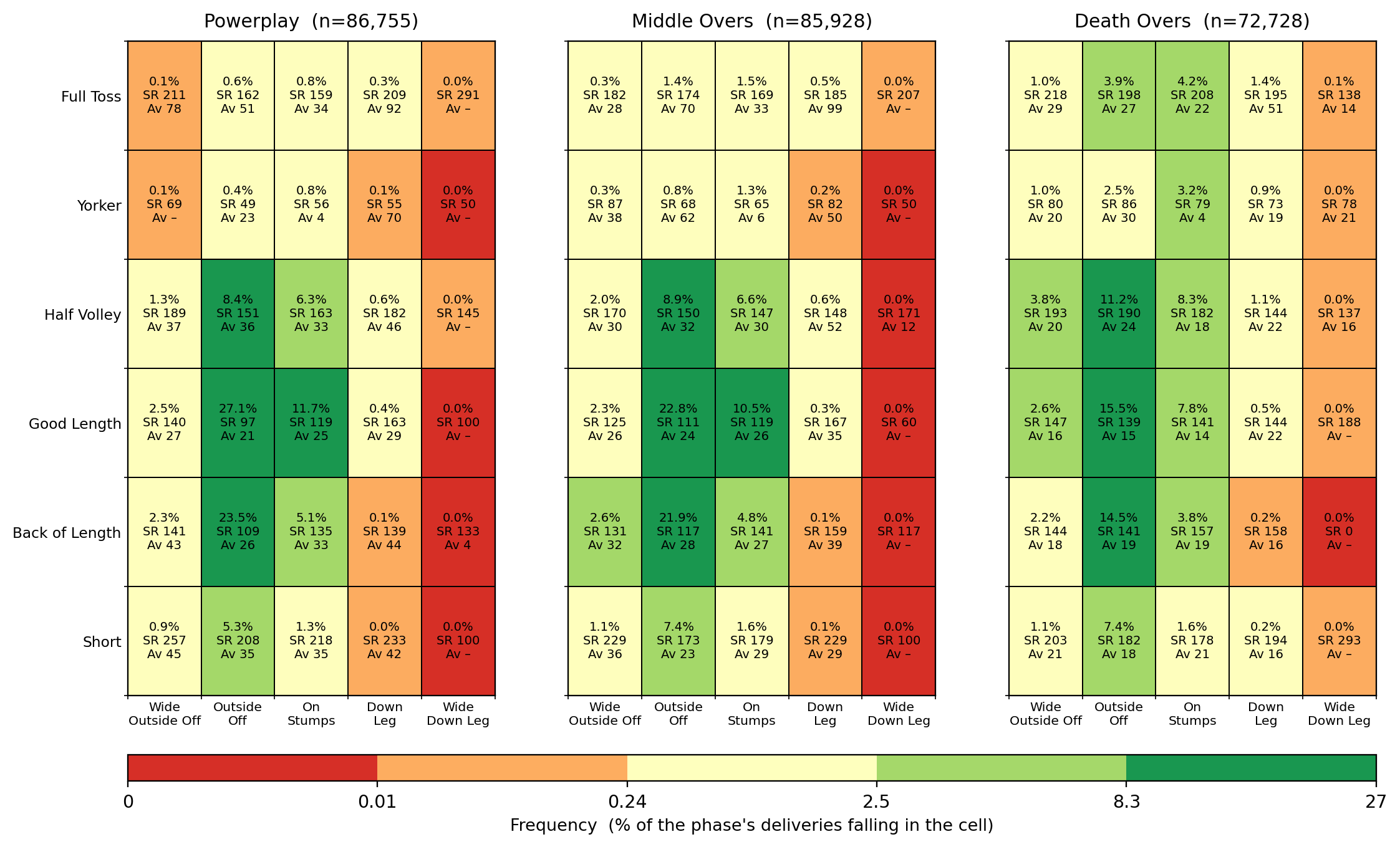}
\caption{Line$\times$length structure for right-hand batters against right-arm pace, by phase. Each
cell reports the frequency (percentage of the phase's deliveries pitching there), strike rate
and batting average in that cell.}
\label{fig:pitchmap}
\end{figure}

\begin{figure}[!ht]
\centering
\includegraphics[width=\linewidth]{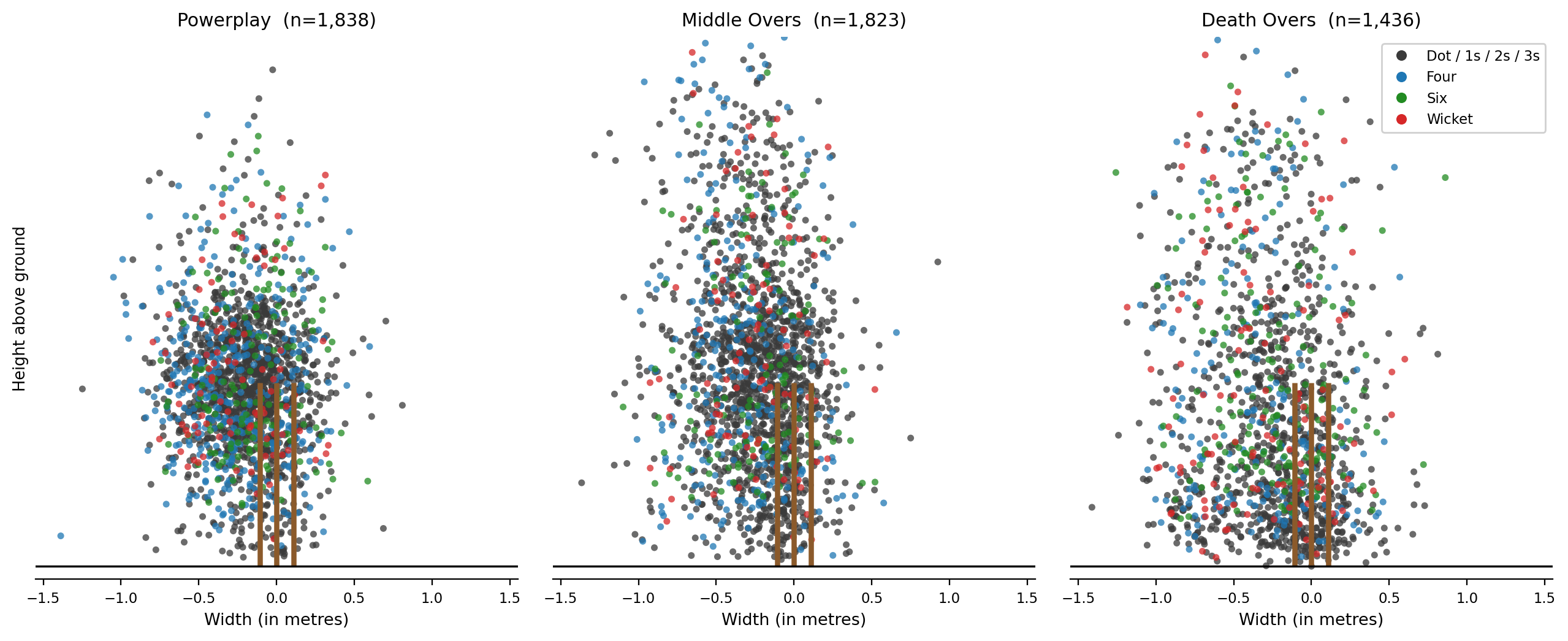}
\caption{Ball impact point (the `beehive' stump view) for right-hand batters against right-arm pace
in IPL 2026, by phase. Each dot is a delivery, placed by its width and
height as it reaches the batter; wickets, fours and sixes are highlighted. The cloud rises and
widens from the powerplay to the death overs.}
\label{fig:beehive}
\end{figure}

The case for line and length is cricketing as much as statistical. Two deliveries identical in every
recorded respect, i.e.\ the same phase, innings, bowler type, venue and match state, are still not
equivalent if one is a good length ball outside off and the other a half volley on the pads. A boundary off the former is a finer stroke than off the latter and, for a fair evaluation, should be
rewarded as such. Symmetrically, being dismissed by a sharp yorker is not the lapse that
missing a straight full toss is, and ideally the two should be debited differently. These `line-and-length'
effects are large, and they shift with the phase and the bowler. Figure~\ref{fig:pitchmap} shows, for
right-hand batters against right-arm pace, how sharply scoring depends on where the ball lands. Against a good-length ball outside off, which is the single most common delivery ($27\%$ of powerplay
balls), the strike rate is only $97$, less than half the $208$ that a short ball on the same line
concedes, and the mix itself moves through the innings, fuller lengths growing far more frequent at
the death overs. Figure~\ref{fig:beehive} tells the same story from the striker's end, plotting where each
delivery reaches the batter. The cloud of impact points\footnote{Impact points, along with the detailed ball-tracking data for IPL 2026, have been obtained from the official IPL website run by BCCI: \href{https://www.iplt20.com/matches/results}{https://www.iplt20.com/matches/results}.}, and the wickets, fours and sixes within it, climbs and spreads as the innings progresses. A model blind to line and length scores all of these deliveries against one bar; conditioning on them, where the data allow, is the refinement Section~\ref{sec:supp:mrae} takes up.

\subsection{Estimation by iterated backfitting}
\label{sec:rae:fit}

The factors in \eqref{eq:rae_model} are correlated: deliveries early (and late) in an innings
are simultaneously in the Powerplay (and Death) phase, at low (and high) wicket counts, and often against
pace, so fitting each against the raw base rate would over-count shared variation. I therefore estimate them by sequential residualisation (backfitting) in the fixed order: scenario $\to$ era $\to$ wicket $\to$ bowler type $\to$ venue $\to$ opposition. Carry a running fitted value $e_b$, initialised at
$e_b=\mu_0$. At any stage, for the current factor $f$ with cells $c$, set
\begin{equation}
\label{eq:rae_ratio}
\hat\gamma^{(f)}_c \;=\; \frac{\sum_{b\in c} r_b}{\sum_{b\in c} e_b},
\end{equation}
then update $e_b\leftarrow e_b\cdot\hat\gamma^{(f)}_{c_f(b)}$ before moving to the next factor. Note that \eqref{eq:rae_ratio} is exactly the multiplicative adjustment that equates each cell's fitted total to its observed total and it has a precise likelihood interpretation.

\begin{proposition}[Cell ratio as a conditional Poisson MLE]
\label{prop:poisson}
Suppose $r_b\mid\bm{x}_b\sim\mathrm{Poisson}(\mu_b)$ independently, with
$\log\mu_b=\log e_b+\log\gamma_{c(b)}$, i.e.\ $\log e_b$ enters as a known offset.
Then the maximum-likelihood estimator of the multiplier $\gamma_c$ for cell $c$ is
$\hat\gamma_c=\big(\sum_{b\in c} r_b\big)\big/\big(\sum_{b\in c} e_b\big)$, which is
\eqref{eq:rae_ratio}.
\end{proposition}
\begin{proof}
The cell-$c$ log-likelihood contribution is
$\sum_{b\in c}\{r_b(\log e_b+\log\gamma_c)-e_b\gamma_c\}+\text{const}$.
Differentiating with respect to $\gamma_c$ and equating it to zero gives $\sum_{b\in c} r_b/\gamma_c-\sum_{b\in c} e_b=0$, hence the stated estimator. The second derivative $-\sum r_b/\gamma_c^2<0$ confirms a maximum.
\end{proof}

Thus each update step given by \eqref{eq:rae_ratio} is one coordinate (block) step of Gauss-Seidel ascent on the Poisson log-likelihood of \eqref{eq:rae_loglin}, holding the other factors as offset. Since the log-likelihood is concave in the log-multipliers, repeating the sweep drives the estimates to that peak, i.e.\ the one best-fitting set of multipliers. On the first pass, each factor is fitted against the factors placed before it, and on every later pass, a factor is re-fitted against all the others (I divide its current multiplier back out of $e_b$ before re-estimating, so no factor counts itself twice). I iterate the sweep till convergence, stopping once the fitted values $e_b$ have effectively stopped changing. The empirical Bayes shrinkage of Section~\ref{sec:rae:shrink} is re-applied each time a factor is fitted.

\begin{remark}[Order independence at convergence]
\label{rem:orderindep}
The maximised log-likelihood is unique, so the fitted expectation and hence the ball-level $\RAE$ does not depend on the order in which the factors are cycled. The ordering (most to least populous) affects only the number of cycles to convergence. Up to the mild empirical-Bayes shrinkage, the fit satisfies the Poisson score identities $\sum_{b\in c} r_b=\sum_{b\in c}\hat\mu_b$ for every cell of every factor, so the residual $\RAE$ is centered within each context cell and not merely in aggregate.
\end{remark}

\subsection{Empirical Bayes shrinkage of the multipliers}
\label{sec:rae:shrink}

Thin cells with very few deliveries (a rare ground in a particular three-year window; a bowler or batter with a short record) yield unstable ratios. Each raw multiplier is therefore shrunk towards the neutral value $1$ by a pseudo-count $\kappa_f$:
\begin{equation}
\label{eq:rae_shrink}
\gamma^{(f)}_c \;=\; \frac{n_c\,\hat\gamma^{(f)}_c + \kappa_f}{n_c + \kappa_f}
\;=\; \omega_c\,\hat\gamma^{(f)}_c + (1-\omega_c)\cdot 1,
\qquad \omega_c=\frac{n_c}{n_c+\kappa_f},
\end{equation}
where $n_c$ is the cell's delivery count. The shrinkage is a convex combination of
the data ratio and unity, with data weight $\omega_c$ increasing in sample size;
$\kappa_f$ is the number of prior deliveries pulling a thin cell back to neutral.

Here $1$ is the neutral no-adjustment value: a multiplier of $1$ leaves the expectation unchanged. Shrinking towards it simply treats a thin, poorly-estimated cell as average until it has enough data to prove otherwise. I estimate one $\kappa_f$ per factor from the data by empirical Bayes, so the degree of shrinkage is itself learned. Model the raw ratios $\hat\gamma^{(f)}_c$ as noisy estimates of true multipliers $\theta_c$ that scatter across cells around a cross-sectional mean of $1$ with variance $\tau_f^2$. Two variances must be separated here. The between-cell variance $\tau_f^2$ is the real spread of the true effects from one cell to the next; this is the signal I want to keep. The within-cell variance is the sampling error in a single cell's estimate, coming only from its finite number of deliveries; this is the noise I want to discount. Because $\hat\gamma^{(f)}_c=\big(\sum_{b\in c}r_b\big)/\big(\sum_{b\in c}e_b\big)$ has a fixed offset in its denominator, this sampling variance is
\begin{equation}
\label{eq:rae_vc}
v_c=\frac{\Var(r) \ n_c}{\big(\sum_{b\in c}e_b\big)^2} .
\end{equation}
I calibrate it to the empirical run variance $\Var(r)$ rather than to a Poisson assumption: runs off the bat are lumpy (point masses concentrated mostly at $0, 1, 2, 4, 6$ and few at $3, 5$) and markedly over-dispersed, with $\Var(r)/\mu_0\approx 2.08$ in the dataset, so a Poisson variance would understate the noise and leave thin cells over-trusted.

I do not observe $\tau_f^2$ directly, since the observed spread of $\hat\gamma^{(f)}_c$ mixes signal and noise. I recover it by the DerSimonian-Laird method of moments \citep{dersimonian1986}, the standard method-of-moments empirical Bayes estimator, which simply subtracts the average sampling variance from the total observed spread, leaving the between-cell signal $\tau_f^2$. The Bayes posterior mean for $\theta_c$ is then the precision-weighted average of the cell's own ratio and the neutral value $1$, with weight
\begin{equation}
\label{eq:rae_weight}
w_c=\frac{\tau_f^2}{\tau_f^2+v_c}=\frac{\text{signal}}{\text{signal}+\text{noise}} .
\end{equation}
So a cell is trusted in proportion to how much of its apparent effect is real rather than sampling noise. Finally, because each delivery's expectation is close to the base rate, $\sum_{b\in c}e_b\approx n_c\mu_0$ and hence $v_c\approx(\Var(r)/\mu_0^2)/n_c$; substituting this into \eqref{eq:rae_weight} collapses the variance weight exactly onto the pseudo-count form $\omega_c=n_c/(n_c+\kappa_f)$ of \eqref{eq:rae_shrink}, with
\begin{equation}
\label{eq:rae_kappa}
\kappa_f=\frac{\Var(r)}{\mu_0^{2} \ \tau_f^{2}}
=\frac{\text{within-cell sampling noise}}{\text{between-cell signal}} .
\end{equation}
Thus a single number $\kappa_f$ (the noise-to-signal ratio for factor $f$), read as an equivalent count of prior deliveries, sets the shrinkage. A factor whose cells differ only weakly beyond noise (small $\tau_f^2$) gives a large $\kappa_f$ and is pulled hard towards $1$, while a factor carrying strong, repeatable variation (large $\tau_f^2$) gives a small $\kappa_f$ and is left close to its raw ratios. A cell holding exactly $\kappa_f$ deliveries sits halfway ($\omega_c=0.5$) between its own estimate and the neutral value. The fitted values (Table~\ref{tab:rae_factors}) recover the expected ordering without it being imposed: the small pace/spin effect (bowler-type) is shrunk most heavily ($\kappa_f \approx 10^3$), while venue, opposition batter and scenario, which carry large heterogeneity, are shrunk the least ($\kappa_f \approx 1.5$-$1.9\times10^2$).

\begin{remark}[Relation to the conjugate Gamma-Poisson posterior]
\label{rem:gamma}
Under a $\mathrm{Gamma}(\kappa,\kappa)$ prior on $\gamma_c$ (prior mean $1$, prior strength $\kappa$) with the Poisson likelihood of Proposition~\ref{prop:poisson}, the posterior mean is $(\sum_{b\in c} r_b+\kappa)/(\sum_{b\in c} e_b+\kappa)$. Estimator \eqref{eq:rae_shrink} replaces the exposure $\sum_{b\in c} e_b$ by the count $n_c$; since $e_b\approx\mu_0\approx 1.24$ runs per ball, $\sum_{b\in c} e_b \approx 1.24\,n_c$, so \eqref{eq:rae_shrink} is a close, slightly more conservative proxy for the exact conjugate posterior mean. I therefore describe it as an empirical Bayes shrinkage estimator rather than an exact Bayes rule.
\end{remark}

\subsection{The opposition overlay and aggregation}
\label{sec:rae:agg}

The single adjustment that specialises the cohort expectation to either role is the opposition factor $\gamma^{\mathrm{opp}}$, estimated symmetrically. For batters' evaluation, a bowler multiplier $\gamma^{\mathrm{bowl}}_{\text{(bowler id)}}$ is co-fit with the context factors in the joint backfitting of Section~\ref{sec:rae:fit} and for bowlers' evaluation, a batter multiplier $\gamma^{\mathrm{bat}}_{\text{(batter id)}}$ is co-fit with the same context factors. The two role-specific fits share an identical context specification, differing only in which opposition identity they carry, and each is applied one at a time.
\begin{itemize}\itemsep1pt
\item Batter's \emph{$\RAE$} multiplies the expectation by the opposition-bowler
      multiplier $\gamma^{\mathrm{bowl}}$, so out-scoring a strong bowler yields a
      larger residual.
\item Bowler's \emph{$\RAE$} multiplies the expectation by the opposition-batter
      multiplier $\gamma^{\mathrm{bat}}$, so containing a strong batter yields a
      larger (more negative) residual.
\end{itemize}
A player's own factor is never applied to his own residual since multiplying a bowler's expectation by his own multiplier would cancel the skill being measured, and likewise for a batter.

The ball-level Runs Above Expected and its player aggregates are finally recorded as
\begin{equation}
\label{eq:rae_ball}
\RAE_b=r_b-\hat\mu(\bm{x}_b),
\qquad
\overline{\RAE}_i=\frac{1}{n_i}\sum_{b:\,i}\RAE_b,
\qquad
\RAE^{\text{total}}_i=\sum_{b:\,i}\RAE_b,
\end{equation}
with standard error $\mathrm{se}_i=\sd_b(\RAE_b)/\sqrt{n_i}$, where the expectation $\hat\mu$ carries the role-appropriate opposition factor. For a batter, $\overline{\RAE}_i$ (denoting $\RAE$/Ball) is the headline rate, centered so that a cohort-average batter scores $\approx 0$, with positive values denoting run production above contextual expectation. For a bowler, $\overline{\RAE}_i$ is the excess runs conceded per ball, so lower (negative) is better: a negative $\RAE$/Ball means the bowler concedes fewer runs than an average bowler would concede in the same situation. The per-ball value used by all downstream signals is $\val_b=\RAE_b$ for batters and $\val_b=-\RAE_b$ for bowlers.

\paragraph{Reading $\RAE$ at different scales.} Because $\RAE$ is defined and additive at the ball level, the same primitive reads coherently at every level of aggregation, with no re-fitting. The expectation is fixed, so a delivery, a knock and a career are all just sums of the identical residual, and each scale answers a distinct question.
\begin{itemize}\itemsep1pt
\item \emph{Single ball-level event} ($\RAE_b$): the value of a single delivery, i.e.\ the unit for
      ball-by-ball diagnostics (determines which deliveries a batter punished, where a bowler was
      beaten etc.).
\item \emph{Per innings or spell} ($\sum_{b\in\text{knock/spell}}\RAE_b$): the value of
      one performance. This is a context-adjusted innings/spell
      measure. For example, a fighting $40$-run innings against quality pace during a collapse can outscore a
      $60$-run innings on a flat deck, because the expectation bar it cleared was higher.
\item \emph{Per-ball average} ($\overline{\RAE}_i$ or $\RAE$/Ball): a player's quality rate,
      the headline per-ball figure, comparable across players irrespective of sample size.
\item \emph{Cumulative} ($\RAE^{\text{total}}_i$):
      the total value contributed by a player over a series, tournament, career or any span of time.
\end{itemize}
The rate and the total $\RAE$ outputs distinctly answer ``how good per ball (or on average)'' versus ``how much value in total'' a player produces. A short-sample specialist can top the former while a durable accumulator tops the latter (Tables~\ref{tab:ipl_bat} and~\ref{tab:iplbowl}). Run prevention, however, is only one of a bowler's two sources of value, and only one side of a dismissal's ledger; the next section prices the second, Dismissal Adjusted Runs.

\section{Dismissal Adjusted Runs (\texorpdfstring{$\DAR$}{DAR})}
\label{sec:dar}

$\RAE$ prices the runs on a delivery, but a delivery that takes a wicket carries value that a per-ball run residual cannot capture: a wicket removes a batter from the remainder of the innings, and the runs he would have gone on to score go with him. From the batting side's view, the side loses whatever those forgone runs exceed what the batters who followed him actually managed; from the bowling side's, the same runs are saved. I call this quantity \emph{Dismissal Adjusted Runs}: the number of runs a dismissal costs the batting side, and equivalently the number it saves the bowling side. It is the framework's second primitive and, like $\RAE$, a currency shared by the two roles. I build it in two steps. Section~\ref{sec:dar:value} constructs a batting-side value function and reads off, for every delivery, the wicket cost $\wc(s)$ it would forgo were a dismissal to occur there, a single non-negative number attached to the match state. Section~\ref{sec:dar:trans} then aggregates this cost twice: once over a player's realised dismissals and once, as a hazard-weighted expectation, over all the balls the batter is party to. The primitive that enters Impact is the difference of the two, which centers the wicket ledger at zero in the same way $\RAE$ centers the run ledger.

A wicket's value is inherently counterfactual, in the sense of ``how would the innings have continued had the batter survived?''. It cannot be read directly off the scorecard, or off any record of events that actually took place, and must instead be estimated against a model of the innings' continuation. The natural model for this estimation is a run-expectancy function or `value' function. To every state of an innings it assigns the runs the batting side can expect to add from that point onward. Given such a function, a wicket needs no separate apparatus to price, rather it is simply the drop in expected remaining runs when a set batter is replaced by the next one in. Valuing game states this way has a long lineage in sport, from the run-expectancy tables of baseball \citep{lindsey1963} to the resource tables of \citet{duckworth1998}, which are themselves a value function of the overs and wickets (together termed as `resources' in DL and DLS methods) a side has in hand.

Such value functions are naturally recursive. The runs expected from a state equal what happens on the next delivery plus the expected value of the state that delivery leads to, averaged over its possible outcomes. This decomposition is a Bellman equation, the identity that organises dynamic programming by tying the value of a state to the values of its successors, and it comes in two forms that must not be confused. The Bellman optimality equation maximises over the actions available at each state and so describes how a side should bat to score the most. It is the object of the dynamic programming studies of optimal scoring rates and batting strategy in limited-overs cricket \citep{clarke1988, preston2000}. The Bellman expectation equation instead holds the policy fixed: it takes the way batters actually score and are dismissed, estimated directly from the data, and returns the expected runs under that behaviour, in the descriptive spirit of models of how scoring evolves through an innings \citep{stevenson2021}. Pricing a wicket calls for the runs an innings would in fact have yielded, not the runs an optimal side could have extracted, so I adopt the expectation form in this paper. Carrying no maximisation, it is a linear recursion, and since the state advances one ball at a time, it is acyclic in terms of balls remaining. So a single backward sweep from the end of the innings solves it exactly, with no optimiser and no fixed point to iterate towards.

\subsection{The batting team value function}
\label{sec:dar:value}
Describe an innings by the state $s=(m,w,i,j)$: $m$ legal balls remain in the innings, the batting side has $w$ remaining wickets in hand, the striker is of quality tier $i$ and the non-striker of tier $j$ (tiers are defined in Section~\ref{sec:dar:trans}), with $e\in\{1,2\}$ denoting the innings. Let
\[
\Vfun_e(m,w,i,j)=\text{expected runs the batting side scores over the remaining $m$ balls from state $s$},
\]
under observed play. The phase of a delivery is a deterministic function of the balls already bowled, given by $\varphi=\varphi(M_0-m)$ with $M_0=120$, so contextual scoring enters through the transition model without enlarging the state.

On the next delivery, the striker (tier $i$) is dismissed with hazard rate $h_{i,\varphi,e}$, otherwise he scores $r\in\mathcal R=\{0,1,2,3,4,5,6\}$ runs with probability $\pi_{r\mid i,\varphi,e}$, where $\sum_{r\in\mathcal R}\pi_{r\mid i,\varphi,e}=1$ (probabilities are conditional on survival of the batter). Odd runs and the end of an over rotate the strike; a dismissal brings in a fresh batter of tier $\iota(w)$ (see details in Section~\ref{sec:dar:trans}). Writing $\sigma(r,m)\in\{0,1\}$ for whether the strike changes ends after the ball (i.e.\ $r$ is odd or the over ends), the value function satisfies the Bellman expectation recursion
\begin{equation}
\label{eq:dar_bellman}
\Vfun_e(m,w,i,j)=
h_{i,\varphi,e}\,\Vfun_e\!\big(m-1,\,w-1,\,\iota(w),\,j\big)
\;+\;(1-h_{i,\varphi,e})\!\!\sum_{r\in\mathcal R}\!\pi_{r\mid i,\varphi,e}\,
\Big[\,r+\Vfun_e\!\big(m-1,\,w,\,i',j'\big)\Big],
\end{equation}
where $(i',j')=(j,i)$ when $\sigma(r,m)=1$ and $(i,j)$ otherwise, and on a dismissal the incoming batter takes strike (its position rotates under $\sigma$ likewise). The recursion closes at the absorbing boundaries
\begin{equation}
\label{eq:dar_bc}
\Vfun_e(0,w,i,j)=0 \quad(\text{balls exhausted}),\qquad
\Vfun_e(m,0,i,j)=0 \quad(\text{batting team gets all out}).
\end{equation}

Equation~\eqref{eq:dar_bellman} is the expectation form introduced above: it evaluates the runs a team scores under its observed scoring and dismissal behaviour, with no maximisation over batting actions. The state space is finite and acyclic when ordered by $m$, so \eqref{eq:dar_bellman} has a unique solution, obtained by a single backward sweep from $m=0$ to $m=M_0$.

\subsection{The transition model}
\label{sec:dar:trans}
The recursion needs, per (tier, phase, innings), a dismissal hazard rate $h$ and a scoring distribution $\pi$, all estimated from the same dataset with the empirical Bayes shrinkage of Section~\ref{sec:rae:shrink}.

\emph{Player rates and quality tiers.} For each batter, I form empirical Bayes shrunk estimates of their scoring rate $\mu_i$ and per-ball dismissal hazard $\eta_i$,
\begin{equation}
\label{eq:dar_rates}
\mu_i=\frac{\sum_{b:\,i}r_b+\kappa_\mu\,\mu_0}{n_i+\kappa_\mu},\qquad
\eta_i=\frac{d_i+\kappa_\eta\,\eta_0}{n_i+\kappa_\eta},
\end{equation}
with $d_i$ denoting the batter's dismissals, $\mu_0$ and $\eta_0$ the corresponding means based on the dataset, and $\kappa_\mu$ and $\kappa_\eta$ the corresponding pseudo-counts in the sense of the empirical Bayes shrinkage theory described in Section~\ref{sec:rae:shrink}. Batters are grouped into $Q=6$ quality tiers by ball-weighted quantiles of $\mu_i$, so a tier is a homogeneous quality band and thin players are pulled to the population before being tiered.

\emph{Scoring and hazard by state.} Pooling deliveries by (tier, phase, innings), the scoring distribution $\pi_{\cdot\mid i,\varphi,e}$ and hazard rate $h_{i,\varphi,e}$ are the corresponding empirical frequencies, each shrunk toward its (phase, innings) pooled distribution by a conjugate (Dirichlet, respectively Beta) pseudo-count so that sparse cells borrow strength. Estimating $\pi$ and $h$ separately by innings lets the second-innings value function reflect chase-state scoring; a finer conditioning on the required run rate is a natural refinement (see Section~\ref{sec:limits}).

\emph{Strike rotation and the incoming batter.} Odd runs and end of an over rotate the strike through $\sigma$. When a wicket falls, the incoming batter's tier $\iota(w)$ is the empirical average tier of batters entering at that position, so the value function unrolls the actual remaining batting order: every subsequent batter is accounted for through the recursion, weighted by how much of the innings remains and how likely each of them is to bat in it.

Because the state is ordered by balls remaining in the innings, one backward pass fills the entire table exactly (Algorithm~\ref{alg:dp}).

\begin{algorithm}[!ht]
\caption{Backward solve of the batting-side value function $\Vfun_e$ (Bellman expectation)}
\label{alg:dp}
\begin{algorithmic}[1]
\State $\Vfun_e(0,\cdot,\cdot,\cdot)\gets 0$;\quad $\Vfun_e(\cdot,0,\cdot,\cdot)\gets 0$
\For{$m=1,\dots,M_0$}
  \State $\varphi \gets \varphi(M_0-m)$;\quad $\text{overEnd}\gets\mathbbm{1}\!\left[(M_0-m+1)\bmod 6=0\right]$
  \For{$w=1,\dots,10$ \textbf{and} tiers $i,j\in\{1,\dots,Q\}$}
    \State $A \gets \sum_{r\in\mathcal R}\pi_{r\mid i,\varphi,e}\big[\,r+\Vfun_e(m-1,\,w,\,i',j')\big]$ \Comment{$(i',j')$ per strike rotation $\sigma(r,m)$}
    \State $B \gets \Vfun_e(m-1,\,w-1,\,\iota(w),\,j)$ \Comment{$0$ if $w=1$ (all out)}
    \State $\Vfun_e(m,w,i,j) \gets h_{i,\varphi,e}\,B + (1-h_{i,\varphi,e})\,A$
  \EndFor
\EndFor
\State \Return $\Vfun_e$
\end{algorithmic}
\end{algorithm}

\paragraph{Per-ball wicket cost.} Consider a delivery in state $s=(m,w,i,j,e)$: a tier-$i$
striker with a tier-$j$ partner, $m$ balls and $w$ wickets remaining in innings $e$. Had the
striker been dismissed here, the batting side would trade its current continuation for the one that
follows a wicket, the ball is consumed and a fresh batter of tier $\iota(w)$ resumes, so the
runs forgone are the drop in the value function,
\begin{equation}
\label{eq:dar_def}
\wc(s)=\Big[\underbrace{\Vfun_e(m,w,i,j)}_{\text{runs if batter survived}}\;-\;\underbrace{\Vfun_e(m-1,w-1,\iota(w),j)}_{\text{runs after the dismissal}}\;+\;\delta_{\text{fresh}}\Big]_{+},
\end{equation}
where $[\,\cdot\,]_+=\max\{\cdot,0\}$. The quantity in brackets is non-negative by construction, since a
set batter is worth at least as much as the fresh one who replaces him, and the positive part
only guards against rare estimation noise. The correction $\delta_{\text{fresh}}\ge0$ accounts
for the incoming batter entering unset: a new batter scores below his settled rate over his
first few deliveries. Writing the empirical `set curve' as $f(\cdot)$, with $f(k)\in(0,1]$ the
fraction of the settled rate realised over the first $k$ balls, and $E=\{1-(1-\eta_i)^{s_i
m}\}/\eta_i$ for the geometric survival expectation of balls the dismissed batter would have
faced (from his hazard $\eta_i$ and strike share $s_i$),
\begin{equation}
\label{eq:delta_fresh}
\delta_{\text{fresh}}=(1-f(E))\,\mu\,\min(E,E_{\max}).
\end{equation}
Note that $\wc(s)$ is a function of the match state, defined on every delivery and not
only on those that happen to take a wicket. Every state is priced on its own terms: the
quality of the striker and his partner, the balls and wickets remaining, the phase and the
innings, so no global `runs per wicket' constant enters anywhere. Dismissing a top-order
batter in the powerplay forgoes many runs; removing a tail-ender in the final over forgoes
almost none, and the same expression returns both.

\paragraph{Realised and expected Dismissal Adjusted Runs.} Let $W_b\in\{0,1\}$ mark whether the striker
was dismissed on delivery $b$, and let $s_b$ be its state. The realised Dismissal Adjusted Runs ($\realDAR$) sums the
cost over a player's actual dismissals, grouped by $\mathcal D_i$, the deliveries on which batter $i$ was
out and $\mathcal W_i$, the deliveries on which bowler $i$ took the wicket,
\begin{equation}
\label{eq:dar_agg}
\realDAR^{\text{bat}}_i=\!\!\sum_{b\in\mathcal D_i}\!\!\wc(s_b),
\qquad
\realDAR^{\text{bowl}}_i=\!\!\sum_{b\in\mathcal W_i}\!\!\wc(s_b) .
\end{equation}
Taken alone, this is a strictly non-negative accumulation, one term per wicket, and it is the
naive quantity one would debit the batter and credit the bowler. It is also what makes the two
roles incomparable if used directly: a mean-zero run ledger ($\RAE$) combined with a one-signed
wicket ledger ($\realDAR$) drives every batter's total down and every bowler's up, an artefact of
mixing a centered and an uncentered quantity rather than a fact about the players
(Section~\ref{sec:val:dar} quantifies the resulting scale gap). The correct comparison is not
against zero but against what the same deliveries would have cost an average player. Let
$\bar h_{\varphi,e}$ be the global dismissal rate per legal ball in phase $\varphi$, innings
$e$ (for evaluations confined to a large enough league, tournament or era, the corresponding dismissal rate can be used instead for finer accuracy). This is defined as a context-level baseline, conditioned on the situation and not on the player being
evaluated, exactly as the cohort expectation $\mu(\bm x_b)$ is in $\RAE$. Summing the cost over
every delivery a player takes part in, weighted by that hazard, gives the expected Dismissal Adjusted
Runs ($\xDAR$). For $\mathcal F_i$ denoting the deliveries batter $i$ faced and $\mathcal B_i$ those bowler $i$ bowled,
\begin{equation}
\label{eq:dar_exp}
\xDAR^{\text{bat}}_i=\!\!\sum_{b\in\mathcal F_i}\!\!\bar h_{\varphi_b,e_b}\,\wc(s_b),
\qquad
\xDAR^{\text{bowl}}_i=\!\!\sum_{b\in\mathcal B_i}\!\!\bar h_{\varphi_b,e_b}\,\wc(s_b) .
\end{equation}
$\xDAR^{\text{bat}}_i$ reads as the wicket value an average batter would be expected to concede
over the situations $i$ actually batted in, and $\xDAR^{\text{bowl}}_i$ the value an average bowler
would be expected to take over the balls $i$ actually bowled.

\paragraph{The centered Dismissal Adjusted Runs.} The primitive that enters the final Impact of a player,
\emph{Dismissal Adjusted Runs} ($\DAR$), is the difference of the two aggregates, realised minus expected.
Because $\sum_{b\in\mathcal F_i}W_b\,\wc(s_b)=\sum_{b\in\mathcal D_i}\wc(s_b)=\realDAR^{\text{bat}}_i$,
it collapses to a single sum over the balls faced,
\begin{equation}
\label{eq:dar_cent}
\DAR^{\text{bat}}_i\;:=\;\realDAR^{\text{bat}}_i-\xDAR^{\text{bat}}_i\;=\;\sum_{b\in\mathcal F_i}\big(W_b-\bar h_{\varphi_b,e_b}\big)\,\wc(s_b),
\end{equation}
and symmetrically for the bowler over $\mathcal B_i$. The factor $\sum_b(W_b-\bar
h_{\varphi_b,e_b})$ is the martingale residual of the dismissal counting process, i.e.\ observed
events minus their compensator, the accumulated hazard (following \cite{barlow1988} and \cite{therneau1990}), so
\eqref{eq:dar_cent} is that residual weighted by the run value of each potential wicket. The
following makes the centering precise.

\begin{proposition}[The centered Dismissal Adjusted Runs is a mean-zero martingale residual]
\label{prop:centre}
If, conditional on the state $s_b$, a batter is dismissed at the league hazard,
$\E[W_b\mid s_b]=\bar h_{\varphi_b,e_b}$, then $\E\big[\DAR^{\text{bat}}_i\big]=\E\big[\realDAR^{\text{bat}}_i-\xDAR^{\text{bat}}_i\big]=0$.
Consequently a batter is charged only for departing from the global (or league) dismissal rate: negative
values (survival beyond expectation) raise Bat Impact, positive values (being dismissed more
often, or in costlier states, than his situations warranted) lower it, and dismissal at exactly
the expected rate is free.
\end{proposition}
\begin{proof}
Taking expectations in \eqref{eq:dar_cent}, $\E[(W_b-\bar h_{\varphi_b,e_b})\wc(s_b)\mid
s_b]=\wc(s_b)\,\E[W_b-\bar h_{\varphi_b,e_b}\mid s_b]=0$ under the stated null; summing over
$\mathcal F_i$ and applying the tower property gives the result. The same argument holds for the
bowler over $\mathcal B_i$.
\end{proof}

\begin{remark}[A transfer of expected runs, conserved system-wide]
\label{rem:transfer}
The wicket cost $\wc(s_b)$ enters the dismissed batter's ledger with a negative sign and the
bowler's with a positive one, so a wicket remains a transfer of expected runs between the sides.
The centering is conserved in aggregate as well: the expected total is the same whether summed over all batters or all bowlers, $\sum_i\xDAR^{\text{bat}}_i=\sum_b\bar h_{\varphi_b,e_b}\,\wc(s_b)=\sum_i\xDAR^{\text{bowl}}_i$, while $\sum_i\realDAR^{\text{bat}}_i=\sum_i\realDAR^{\text{bowl}}_i$ is the realised total. The framework therefore
stays balanced globally while judging each player against the situations they personally faced,
the same property $\RAE$ possesses.
\end{remark}

\section{Impact: batting and bowling on the same runs scale}
\label{sec:composite}

With two primitives in hand, i.e.\ $\RAE$ for excess runs on every ball and the centered Dismissal Adjusted
Runs ($\DAR=\realDAR-\xDAR$) for excess wicket value, a player's Impact is his total contribution across
both channels. The construction is symmetric between the roles and differs only in sign: a
batter and a bowler each add value by scoring or saving runs, and transfer value at dismissals
relative to what those dismissals were expected to be worth.

\paragraph{Batting Impact.} A batter adds runs above expectation ball by ball ($+\RAE$) and
carries the centered wicket ledger:
\begin{equation}
\label{eq:bat_impact}
\OmegaBat_i=\sum_{b\in\mathcal F_i}\RAE_b\;-\;\DAR^{\text{bat}}_i .
\end{equation}
The first term rewards how he scored, a per-ball residual against a contextual par, so a
below-par passage accrues little or negative $\RAE$ however long it lasts. The second term,
expanded as $\realDAR^{\text{bat}}_i-\xDAR^{\text{bat}}_i$, credits the wicket value an average batter
would have surrendered over his situations and debits what he actually surrendered. A batter who
outlasts the global (or league) dismissal rate is rewarded for it; one who is dismissed early or cheaply,
relative to the states he batted in, is charged the excess. This differs from a raw wicket
penalty in an important way. Getting out is the expected end of every innings, so the metric does not tax a batter simply for being dismissed, rather it taxes only dismissals beyond expectation,
and pays a survival dividend below it. Occupation of the crease is close to neutral rather than
free, since the survival credit it earns is offset by the weak $\RAE$ that slow scoring
produces, while sustained batting at or above par is rewarded through both channels at once.

\paragraph{Bowling Impact.} Symmetrically, a bowler saves runs ball by ball ($-\RAE$) and
carries the centered wicket credit:
\begin{equation}
\label{eq:bowl_impact}
\OmegaBowl_i=-\sum_{b\in\mathcal B_i}\RAE_b\;+\;\DAR^{\text{bowl}}_i .
\end{equation}
Here the wicket term credits a bowler for taking wickets beyond what his overs would ordinarily
yield and neither rewards nor penalises one who takes exactly the expected number. This is the
correction that a raw wicket count misses: a high-volume containment bowler accumulates wickets
simply by bowling many deliveries, and pricing those against the expected haul strips out the
inflation while leaving the strike bowlers who beat expectation ahead. Run prevention and this
net wicket credit are the two separate sources of bowling value, and both are measured in runs,
so they add directly. Their relative weight is not a tuning parameter but a consequence of the value function and the hazard baseline.

\begin{remark}[One centered scale, comparable across roles]
\label{rem:onescale}
Both \eqref{eq:bat_impact} and \eqref{eq:bowl_impact} are sums of two mean-zero residuals in
runs above expectation. A batter's Impact and a bowler's Impact are therefore readily comparable: a value of $+300$ denotes the same $300$ runs of net contribution
whichever role produced it, and an average performer of either role sits near the common origin. Under the naive aggregate
$\sum\RAE\mp\realDAR$, a one-signed wicket total added to a centered run total pushes the roles onto
disjoint ranges; subtracting the expectation $\xDAR$ restores a shared scale, as
Section~\ref{sec:val:dar} documents. Impact may be read as a career or tournament total, which is
longevity-aware because both the run value and the survival credit accumulate with volume, or as
a per-ball \emph{rate} $\OmegaBat_i/n_i$; the two answer peak-quality and total-contribution
questions respectively.
\end{remark}

Given the value table of Algorithm~\ref{alg:dp} and the hazard baseline $\bar h_{\varphi,e}$,
both aggregates assemble in a single linear pass over the deliveries
(Algorithm~\ref{alg:impact}): each ball adds its $\RAE$ to the striker and its negative to the
bowler, adds the hazard-weighted wicket cost $\bar h_{\varphi_b,e_b}\wc(s_b)$ to both players'
expected ledgers, and, on a dismissal, adds the realised cost $\wc(s_b)$ as well.

\begin{algorithm}[!ht]
\caption{Assembly of batting and bowling Impact in one pass over deliveries}
\label{alg:impact}
\begin{algorithmic}[1]
\State $R^{\text{bat}},R^{\text{bowl}},\realDAR^{\text{bat}},\realDAR^{\text{bowl}},\xDAR^{\text{bat}},\xDAR^{\text{bowl}}\gets 0$ for all players
\For{each legal delivery $b$ with striker $p$, bowler $q$, state $s_b$}
  \State $R^{\text{bat}}_p \mathrel{+}= \RAE_b$;\quad $R^{\text{bowl}}_q \mathrel{-}= \RAE_b$ 
  \State $\xDAR^{\text{bat}}_p \mathrel{+}= \bar h_{\varphi_b,e_b}\,\wc(s_b)$;\quad $\xDAR^{\text{bowl}}_q \mathrel{+}= \bar h_{\varphi_b,e_b}\,\wc(s_b)$ 
  \State \textbf{if} a wicket falls ($W_b=1$): $\realDAR^{\text{bat}}_{d(b)} \mathrel{+}= \wc(s_b)$,\ \ $\realDAR^{\text{bowl}}_q \mathrel{+}= \wc(s_b)$ 
\EndFor
\State $\OmegaBat_p \gets R^{\text{bat}}_p-(\realDAR^{\text{bat}}_p-\xDAR^{\text{bat}}_p)$;\quad $\OmegaBowl_q \gets R^{\text{bowl}}_q+(\realDAR^{\text{bowl}}_q-\xDAR^{\text{bowl}}_q)$
\State \Return $\{\OmegaBat_p\},\{\OmegaBowl_q\}$
\end{algorithmic}
\end{algorithm}

\section{Results and discussions}
\label{sec:validation}

This section documents the fitted framework on the entire dataset: the structure of the expected runs model (Section~\ref{sec:val:rae}), the distribution of $\RAE$ (Section~\ref{sec:val:dist}), the state dependence of Dismissal Adjusted Runs (Section~\ref{sec:val:dar}), the correlation structure and convergent validity (Section~\ref{sec:val:corr}), and the face-validity leaderboards from an IPL case study
(Section~\ref{sec:val:lead}). All quantities are computed by the production estimators of
Sections~\ref{sec:rae}-\ref{sec:composite}.

\subsection{Structure of the expected runs model}
\label{sec:val:rae}

The grand mean is $\mu_0=1.2425$ runs per legal delivery. Table~\ref{tab:rae_factors} summarises the seven fitted factors and their shrinkage pseudo-counts; Table~\ref{tab:rae_cells} reports the interpretable scenario, wicket and bowler-type multipliers, and Figure~\ref{fig:rae_factors} plots the era trend.

\begin{finding}[Death overs premium, Powerplay discount]
\label{find:scenario}
Conditional on all other factors, expected runs in the Death overs run $+18.8\%$
(innings~1) and $+13.2\%$ (innings~2) above base, while the Powerplay sits
$0$-$4\%$ below base. The strike rate premium of the Powerplay visible in
raw data is therefore largely a phase artefact: because the scenario factor raises
the Death overs bar, a run scored (or saved) there earns proportionately less (or
more) $\RAE$ than the same run in the lower-scoring Powerplay.
\end{finding}

\begin{table}[!ht]
\centering
\caption{The fitted $\RAE$ factor model. Multiplier ranges for the low-cardinality factors are over interpretable cell aggregates; for the high-cardinality venue and opposition factors, they are raw shrunk-cell extrema. $\tau_f^2$ is the empirical Bayes between-cell variance and $\kappa_f=\Var(r)/(\mu_0^2\tau_f^2)$ the resulting data estimated shrinkage pseudo-count \eqref{eq:rae_kappa}. The batter and bowler factors are the symmetric opposition multipliers of Section~\ref{sec:rae:agg}, applied one at a time.}
\label{tab:rae_factors}
\small
\begin{tabular}{lrcrr}
\toprule
Factor (cell key) & No. of cells & Multiplier range & $\tau_f^2$ & $\kappa_f$ \\
\midrule
Scenario\ (phase $\times$ innings)         & 8       & $0.96$ - $1.27$ & $0.0093$ & 180 \\
Era\ (year $\times$ phase)                 & 69      & $0.96$ - $1.14$ & $0.0016$ & 1{,}043 \\
Wicket\ (phase $\times$ wickets fallen)    & 33      & $0.66$ - $1.16$ & $0.0034$ & 485 \\
Bowler type\ (phase $\times$ sub-type)     & 21      & $0.97$ - $1.03$ & $0.0013$ & 1{,}249 \\
Venue\ (ground $\times$ 3-yr bucket)       & 1{,}268 & $0.17$ - $1.36$ & $0.0054$ & 311 \\
Opposition bowler\ (identity)              & 7{,}282 & $0.77$ - $1.18$ & $0.0061$ & 273 \\
Opposition batter\ (identity)              & 9{,}740 & $0.50$ - $1.87$ & $0.0249$ & 67 \\
\midrule
\multicolumn{5}{l}{\footnotesize Base rate $\mu_0 = 1.2425$ runs / legal ball; $\Var(r)=2.578$.}\\
\bottomrule
\end{tabular}
\end{table}

\begin{table}[!ht]
\centering
\caption{Interpretable context multipliers. Values above $1$ raise, and below $1$
lower, the cohort expected runs relative to base.}
\label{tab:rae_cells}
\small
\begin{tabular}{lcc@{\hskip 2.4em}lc}
\toprule
\multicolumn{3}{c}{Scenario $\gamma^{\mathrm{scen}}$} &
\multicolumn{2}{c}{Bowler type $\gamma^{\mathrm{type}}$} \\
\cmidrule(r){1-3}\cmidrule(r){4-5}
Phase & Inns 1 & Inns 2 & Sub-type & Multiplier\\
\midrule
Powerplay    & 0.959 & 0.999 & Right-arm pace & 1.028 \\
Middle overs & 0.968 & 0.975 & Left-arm pace  & 0.998 \\
Death overs  & 1.188 & 1.132 & Leg-spin       & 1.005 \\
             &       &       & Off-spin       & 0.989 \\
             &       &       & Left-arm spin  & 0.968 \\
\bottomrule
\end{tabular}
\end{table}

\begin{table}[!ht]
\centering
\caption{Wicket multiplier $\gamma^{\mathrm{wkt}}$ by phase and number of wickets fallen.
Values above $1$ raise, and below $1$ lower, the cohort expected runs relative to base.}
\label{tab:rae_wkt}
\small
\setlength{\tabcolsep}{4.5pt}
\begin{tabular}{lccccccccccc}
\toprule
Wickets fallen & 0 & 1 & 2 & 3 & 4 & 5 & 6 & 7 & 8 & 9 & 10 \\
\midrule
Powerplay    & 1.047 & 1.025 & 0.920 & 0.825 & 0.746 & 0.716 & 0.851 & 0.911 & 1.008 & 0.984 & -- \\
Middle overs & 1.136 & 1.081 & 1.036 & 0.988 & 0.938 & 0.889 & 0.840 & 0.770 & 0.735 & 0.657 & 0.977 \\
Death overs  & 0.986 & 1.164 & 1.150 & 1.117 & 1.083 & 1.028 & 0.975 & 0.895 & 0.824 & 0.692 & 0.751 \\
\bottomrule
\end{tabular}
\end{table}

\begin{figure}[!ht]
\centering
\includegraphics[width=0.95\linewidth]{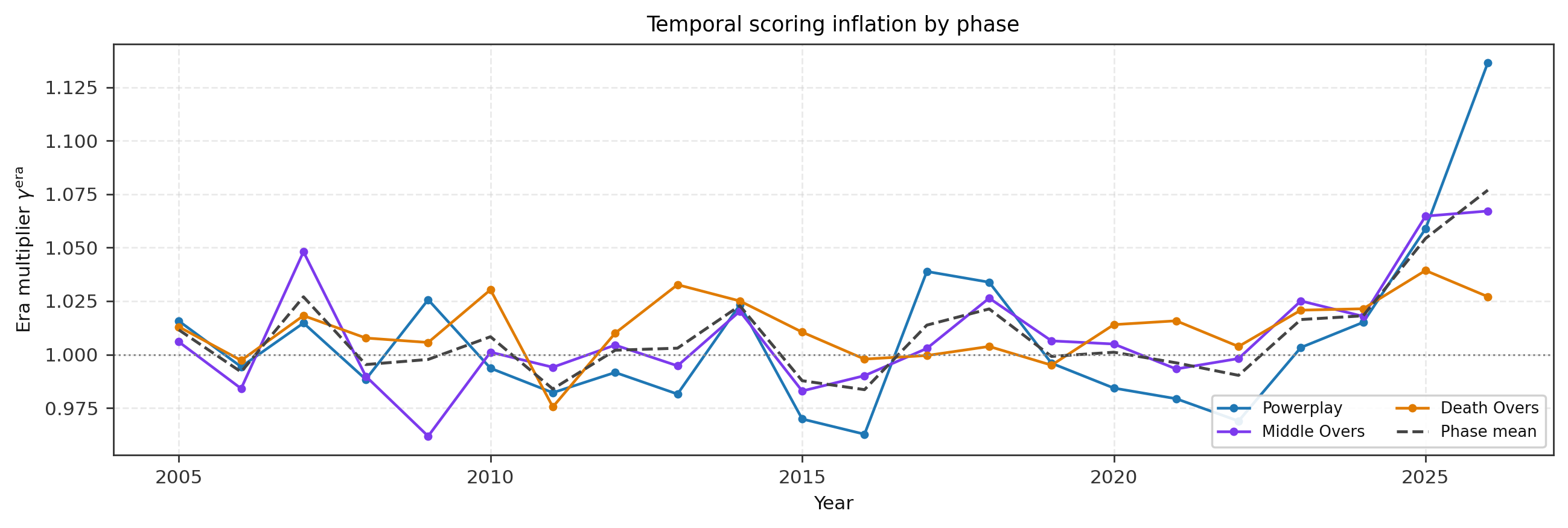}
\caption{The era multiplier $\gamma^{\mathrm{era}}$ by year, shown separately for
each match phase together with the phase mean (dashed). Scoring conditions are
broadly stationary across the sample with a mild inflation in the most recent
seasons. The phase~$\times$~innings scenario multipliers are listed exactly in
Table~\ref{tab:rae_cells}.}
\label{fig:rae_factors}
\end{figure}

\begin{finding}[Wicket pressure is phase-dependent]
\label{find:wicket}
The wicket factor (Table~\ref{tab:rae_wkt}) is sharp where
deliveries are plentiful and shrunk toward neutral where they are not. In the middle overs it
declines almost monotonically from $1.136$ (0 down) to $0.657$ (9 down), so each further
wicket lowers cohort expected runs. The powerplay shows the same early decline
($1.047\to0.716$ across the first five wickets), while its deep-collapse cells (6 or more down
inside the first six overs, a rare event) carry little information and are pulled back towards
$1$. The death overs invert the top of the gradient: one wicket down $(1.164)$ scores above none
$(0.986)$, consistent with a side that accelerates once a set batter falls, before the decline
resumes. Banding the counts averages this phase-dependent
structure away, while conditioning on the exact count retains it, with the empirical Bayes shrinkage
absorbing the sparse cells automatically.
\end{finding}

\begin{finding}[Pace is more expensive than spin]
\label{find:pace}
Averaged over phases, right-arm pacers concede the most $(1.028)$ and finger spinners the
least (off-spin $0.989$, left-arm spin $0.968$), with left-arm pace $(0.998)$ and
leg-spin $(1.005)$ sitting near base: $\approx 6\%$ expected runs gap between
right-arm pace and left-arm finger spin, absorbed by the bowler-type factor so that
$\RAE$ neither rewards batters for merely facing more pace nor penalises bowlers for
bowling it.
\end{finding}

\subsection{Distribution of \texorpdfstring{$\RAE$}{RAE}}
\label{sec:val:dist}

Figure~\ref{fig:rae_dist} shows the ball-level and batter-level $\RAE$ distributions.

\begin{figure}[H]
\centering
\includegraphics[width=0.92\linewidth]{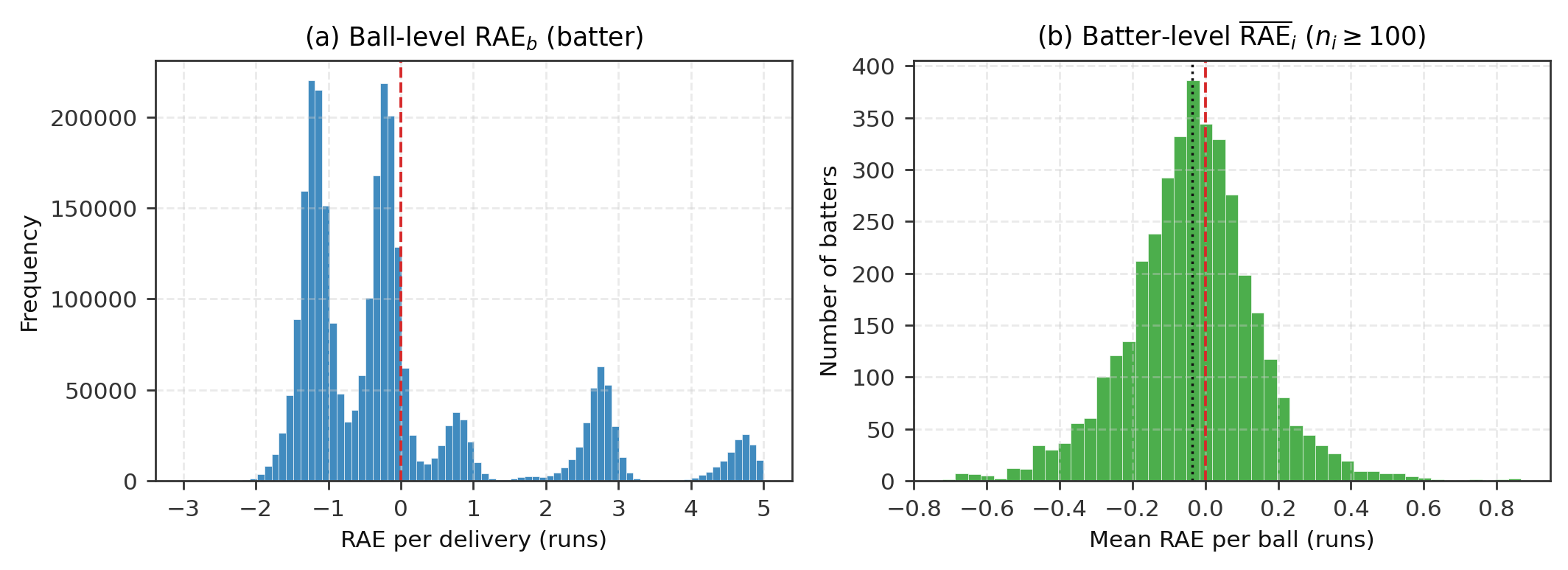}
\caption{(a) Ball-level batter $\RAE_b$ (truncated to $[-3,5]$ for display): a
right-skewed point-mass mixture centered at zero, the long right tail corresponding
to boundaries. (b) Batter-level $\overline{\RAE}_i$ for batters with $n_i\ge100$;
the dotted line marks the cohort mean.}
\label{fig:rae_dist}
\end{figure}

\begin{remark}[The residual is centered]
\label{rem:centred}
Across all $2{,}727{,}420$ legal deliveries, the mean ball-level batter residual is
$\overline{\RAE}_b=-0.0001$ runs, i.e.\ four orders of magnitude below the base
rate $\mu_0=1.2425$, with standard deviation $1.588$ and right-skew $1.58$; the
bowler residual is likewise centered $(-0.0005$, s.d.\ $1.582)$. This centering is a
property of the converged fit, whose Poisson score identities (Remark~\ref{rem:orderindep})
equate fitted and observed totals within every context cell; the skew and excess
mass at $\{-\mu_0,\,4,\,6\}$ reflect the discrete run distribution rather than model
misfit.
\end{remark}

\subsection{The wicket cost and restoration due to centering}
\label{sec:val:dar}
Across all $147{,}797$ bowler-credited dismissals in the dataset, the per-ball wicket cost
$\wc(s)$ evaluated at the dismissal state averages $2.65$ runs (median $2.56$, $90$th percentile
$5.34$). A dismissal, therefore, costs a batting side, on average, only a couple of
runs beyond the ball on which it falls, because most wickets remove middling batters in the
middle or late overs where the marginal cost of one more is small; the tail of high-value
dismissals (a set top-order batter removed early) reaches $15$ runs. This spread is produced
entirely by the state: the value function assigns every wicket its own worth from the balls and
wickets remaining, the phase, the innings, and the quality of the batter dismissed and their
partner.

The cost is highest in the powerplay and roughly flat thereafter (Table~\ref{tab:dar_over}). A
powerplay dismissal costs $\approx23\%$ more than a death-overs dismissal does, because removing a batter
early exposes the rest of the order to more of the innings. The gradient is not monotone through
the innings, however the middle overs are the cheapest, and the death overs recover slightly
because the few balls a late wicket forgoes are the highest-scoring of the innings, an effect
that partly offsets their small number.

\begin{table}[!ht]
\centering
\caption{Mean wicket cost $\wc(s)$ (runs) by over band, over all bowler-credited dismissals.
The cost is largest in the powerplay, dips through the middle overs, and recovers marginally at
the death, so it is not monotone in the balls remaining.}
\label{tab:dar_over}
\small
\begin{tabular}{lcccc}
\toprule
Over band & 1-6 (Powerplay) & 7-10 (Middle 1) & 11-15 (Middle 2) & 16-20 (Death) \\
\midrule
Mean $\wc(s)$ (runs) & $3.10$ & $2.49$ & $2.42$ & $2.52$ \\
\bottomrule
\end{tabular}
\end{table}

\begin{finding}[Wicket value is situational, not constant]
\label{find:dar_state}
No single `runs per wicket' number describes the data: the wicket cost ranges from near zero (a
tail-ender removed on the last ball of a settled match) to approximately $15.5$ runs (a set top-order batter
dismissed early), with a powerplay-to-death phase profile of $3.10\to2.52$ runs
(Table~\ref{tab:dar_over}) and a further, orthogonal dependence on the quality of the batter
dismissed. This is the behaviour a run value of a wicket should have, and it is why
the value function, rather than a flat exchange rate, is required.
\end{finding}

The realised total of these costs is what a naive symmetric aggregate would debit each batter
and credit each bowler. Doing so is what breaks role comparability. Figure~\ref{fig:centering}(a)
shows the two Impact distributions under the absolute construction $\sum\RAE\mp\realDAR$ across
qualifying IPL players ($137$ batters at $\ge500$ balls, $207$ bowlers at $\ge300$ balls). On this subset of players, Batting
Impact averages $-193$ runs (median $-159$, maximum $210$) while Bowling Impact averages $+225$
(median $+123$, maximum $1799$). The distributions barely overlap, and the ratio of the leading
bowler to the leading batter is roughly $8.6$ to $1$.

Subtracting the expected Dismissal Adjusted Runs ($\xDAR$) removes the artefact. Panel (b) shows the same
players under the centered aggregate of \eqref{eq:bat_impact}-\eqref{eq:bowl_impact}: batting
Impact now averages $+144$ (median $+87$, maximum $1022$) and bowling Impact $-15$ (median $-34$,
maximum $924$), so the two roles occupy a common range around the origin and their leaders differ
by less than $10\%$. The shift is not a rescaling but a re-centering: an average bowler who takes
the wickets his volume ordinarily yields, previously flattered by a large positive wicket total,
now sits near zero, while the strike bowlers who beat expectation stay high.

\begin{figure}[H]
\centering
\includegraphics[width=0.92\linewidth]{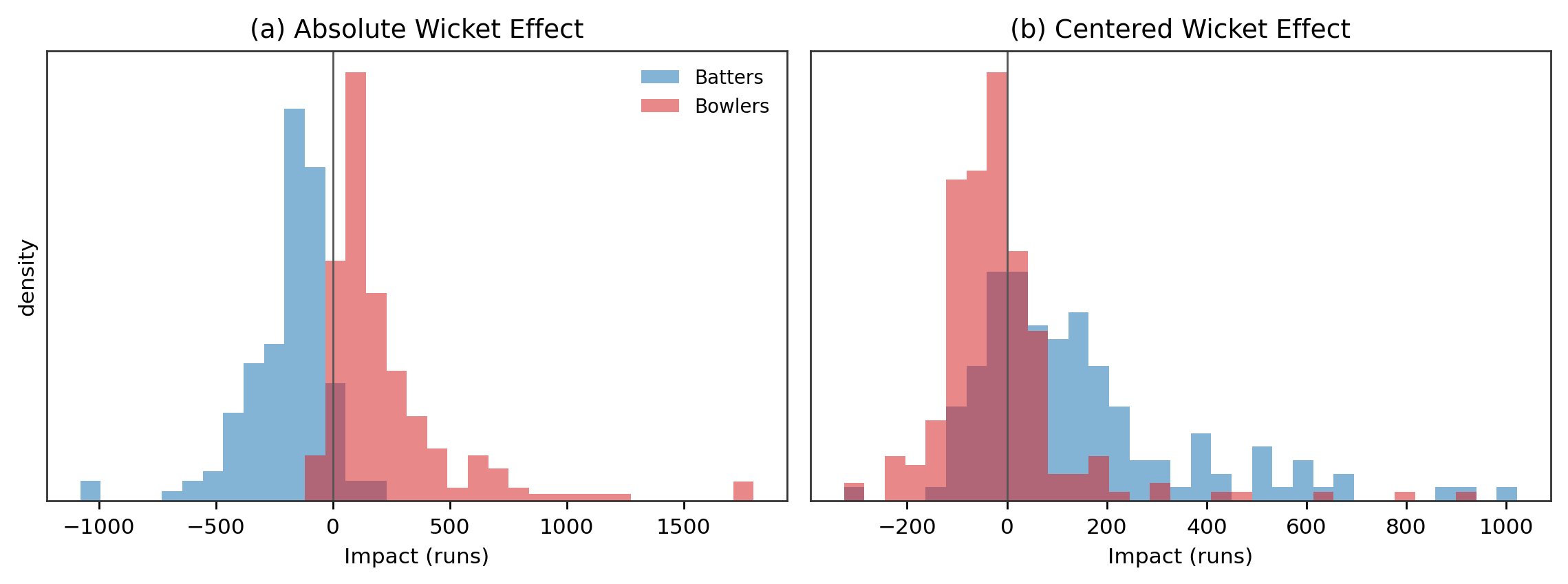}
\caption{Batting (blue) and bowling (red) Impact distributions over qualifying IPL players.
\textbf{(a)} The absolute construction $\sum\RAE\mp\realDAR$ places the roles on disjoint ranges: a
mean-zero run ledger combined with a one-signed wicket ledger. \textbf{(b)} The centered
construction $\sum\RAE\mp\DAR$ returns both roles to a shared, comparable scale about the
origin. Densities; the vertical line marks zero.}
\label{fig:centering}
\end{figure}

\begin{finding}[Centering the wicket ledger restores role comparability]
\label{find:centering}
Pricing wickets in runs is necessary but not sufficient for a comparable batting and bowling
scale: $\realDAR$ is a one-signed accumulation, and adding it to the mean-zero
$\RAE$ leaves batters and bowlers on ranges that differ by nearly an order of magnitude
(Figure~\ref{fig:centering}a). Subtracting each player's $\xDAR$, i.e.\ the martingale
compensator of Proposition~\ref{prop:centre}, returns both roles to mean-zero residuals on one
runs scale (Figure~\ref{fig:centering}b), so that a batter's and a bowler's Impact share a
common centre as well as a common unit.
\end{finding}

\subsection{Correlation structure and convergent validity}
\label{sec:val:corr}
Since the framework rests on two primitives, the natural validity checks are that (i)
each primitive agrees with the conventional statistic it is meant to context-adjust
(convergent validity), and (ii) the two sources of bowling value are distinct
from each other (discriminant validity). Tables~\ref{tab:bat_corr}-\ref{tab:bowl_corr}
report Pearson correlations across qualifying IPL players ($n=137$ batters at
$\ge500$ balls, $n=207$ bowlers at $\ge300$ balls).

\begin{table}[H]
\centering
\begin{minipage}{0.46\linewidth}
\centering
\captionsetup{justification=centering}
\caption{Batting Impact correlations}
\label{tab:bat_corr}
\small
\begin{tabular}{lccc}
\toprule
 & Impact & $\RAE$/Ball & Strike Rate \\
\midrule
Impact       & $1.00$ &        &       \\
$\RAE$/Ball  & $0.65$ & $1.00$ &       \\
Strike Rate           & $0.54$ & $0.82$ & $1.00$\\
\bottomrule
\end{tabular}
\end{minipage}\hfill
\begin{minipage}{0.52\linewidth}
\centering
\captionsetup{justification=centering}
\caption{Bowling Impact correlations}
\label{tab:bowl_corr}
\small
\setlength{\tabcolsep}{4pt}
\begin{tabular}{lcccc}
\toprule
 & Impact & $-\RAE$/Ball & $\DAR$/Ball & Economy \\
\midrule
Impact     & $1.00$ &        &        &       \\
$-\RAE$/Ball  & $0.73$ & $1.00$ &        &       \\
$\DAR$/Ball    & $0.13$ & $0.09$ & $1.00$ &       \\
Economy             & $-0.46$& $-0.68$& $-0.07$ & $1.00$\\
\bottomrule
\end{tabular}
\end{minipage}
\end{table}

\begin{finding}[$\RAE$ is essentially a context-adjusted run rate]
\label{find:convergent}
$\RAE/\text{Ball}$ correlates $0.82$ with the batter's raw strike rate and the bowler's runs saved
correlate $-0.68$ with economy. Each primitive tracks the conventional rate it refines,
confirming that $\RAE$ measures scoring rate or run-prevention rate. The correlations are strong but far from one, and the gap is exactly the contextual
adjustment (era, phase, wicket pressure, venue, opposition) that any raw performance metric ignores. Batting Impact correlates only $0.54$ with strike rate, because it also
carries the centered wicket ledger, through survival credit and dismissal charge, which strike
rate cannot see.
\end{finding}

\begin{finding}[Wicket-taking and run prevention are nearly orthogonal]
\label{find:two_sources}
On the bowling side, the two value sources are almost uncorrelated. The net Dismissal Adjusted Runs per
ball correlates only $0.09$ with runs saved per ball. Wicket-taking is therefore a separate dimension of
bowling value, not a by-product of containment: a miserly bowler need not
be a strike bowler (e.g. Jasprit Bumrah), nor the reverse (e.g. Arshdeep Singh). This justifies Bowling Impact carrying both terms in \eqref{eq:bowl_impact}: economy alone (which correlates $-0.46$ with Bowling Impact) would misjudge the wicket-hunters, and neither term subsumes the other. The more modest correlation of net Dismissal Adjusted Runs with Bowling Impact itself ($0.13$) is a direct consequence of centering. Once the expected haul is removed, most bowlers contribute little through wickets, and the minority of strike bowlers does the separating.
\end{finding}

\subsection{Leaderboards based on the Indian Premier League}
\label{sec:val:lead}
The framework's rankings should recover cricketing consensus while remaining reproducible
from the ball-by-ball record. Among the 18 domestic and franchise leagues across Australia, Bangladesh, England, India, Ireland, Nepal, New Zealand, Pakistan, South Africa, Sri Lanka and West Indies available in our dataset, the Indian Premier League (IPL) records the largest sample size, with $295{,}629$ deliveries ($285{,}761$ of them legal) spanning $18$ consecutive seasons ($2008$-$2026$), involving $769$ batters and $603$ bowlers.

I report IPL leaderboards ($2008$-present) as the primary case study, both as per-ball rate and as cumulative Impacts, and the two answer different questions (peak rate versus sustained contribution).

\begin{table}[!ht]
\centering
\caption{Top batters in IPL sorted by (a) per-ball rate $\RAE$/Ball; (b) cumulative Bat Impact
$\OmegaBat=\RAE-\DAR^{\text{bat}}$ (min $500$ balls).}
\label{tab:ipl_bat}
\begin{subtable}{0.18\linewidth}
\centering
\subcaption{By $\RAE$/Ball}
\label{tab:ipl_rae}
\scriptsize
\begin{tabular}{rlrrr}
\toprule
Rank & Batter & $\RAE$/Ball & Runs & SR \\
\midrule
1  & Andre Russell     & $0.368$ & $2651$ & 175.5 \\
2  & Virender Sehwag   & $0.341$ & $2728$ & 155.4 \\
3  & Sunil Narine      & $0.278$ & $1820$ & 165.2 \\
4  & Phil Salt         & $0.276$ & $1258$ & 175.7 \\
5  & Abhishek Sharma   & $0.271$ & $2379$ & 171.3 \\
6  & Glenn Maxwell     & $0.257$ & $2819$ & 155.1 \\
7  & Sanath Jayasuriya & $0.238$ & $\ 768$ & 144.4 \\
8  & Travis Head       & $0.221$ & $1556$ & 170.1 \\
9  & Nicholas Pooran   & $0.215$ & $2527$ & 164.0 \\
10 & Tim David         & $0.210$ & $1151$ & 176.5 \\
11 & AB de Villiers    & $0.209$ & $5162$ & 151.9 \\
12 & Chris Gayle       & $0.204$ & $4965$ & 149.2 \\
\bottomrule
\end{tabular}
\end{subtable}\hfill
\begin{subtable}{0.52\linewidth}
\centering
\subcaption{By cumulative Bat Impact $\OmegaBat$}
\label{tab:ipl_imp}
\scriptsize
\setlength{\tabcolsep}{3pt}
\begin{tabular}{rlrrrrrr}
\toprule
Rank & Batter & $\OmegaBat$ & $\RAE$ & $\realDAR$ & $\xDAR$ & Runs & SR \\
\midrule
1  & AB de Villiers    & $1022$ & $713$  & $514$  & $\ 824$ & $5162$ & 151.9 \\
2  & David Warner      & $\ 911$& $540$  & $762$  & $1133$  & $6565$ & 139.6 \\
3  & Chris Gayle       & $\ 891$& $684$  & $568$  & $\ 775$ & $4965$ & 149.2 \\
4  & KL Rahul          & $\ 672$& $204$  & $550$  & $1017$  & $5815$ & 139.1 \\
5  & Andre Russell     & $\ 667$& $543$  & $333$  & $\ 457$ & $2651$ & 175.5 \\
6  & Suryakumar Yadav  & $\ 619$& $424$  & $460$  & $\ 655$ & $4581$ & 149.1 \\
7  & MS Dhoni          & $\ 601$& $\ 68$ & $535$  & $1069$  & $5439$ & 137.5 \\
8  & Virat Kohli       & $\ 601$& $\ 11$ & $1013$ & $1604$  & $9336$ & 134.8 \\
9  & Suresh Raina      & $\ 584$& $371$  & $678$  & $\ 891$ & $5528$ & 136.7 \\
10 & Jos Buttler       & $\ 545$& $312$  & $508$  & $\ 741$ & $4646$ & 149.7 \\
11 & Virender Sehwag   & $\ 523$& $598$  & $509$  & $\ 434$ & $2728$ & 155.4 \\
12 & Shane Watson      & $\ 511$& $390$  & $538$  & $\ 659$ & $3874$ & 138.0 \\
\bottomrule
\end{tabular}
\end{subtable}
\end{table}

\begin{table}[!ht]
\centering
\caption{Top bowlers in IPL sorted by (a) per-ball run prevention $-\RAE$/Ball; (b) cumulative
Bowl Impact $\OmegaBowl=-\sum\RAE+\DAR^{\text{bowl}}$ (min $300$ balls).}
\label{tab:iplbowl}
\begin{subtable}{0.37\linewidth}
\centering
\subcaption{By $-\RAE$/Ball}
\label{tab:iplbowl_rate}
\scriptsize
\setlength{\tabcolsep}{3pt}
\begin{tabular}{rlrrr}
\toprule
Rank & Bowler & $-\RAE$/Ball & Wkts & Econ \\
\midrule
1  & Matheesha Pathirana & $0.292$ & \ 47 & 7.73 \\
2  & Jasprit Bumrah      & $0.286$ & 187 & 7.07 \\
3  & Sunil Narine        & $0.227$ & 207 & 6.66 \\
4  & Jofra Archer        & $0.217$ & \ 84 & 7.78 \\
5  & Anil Kumble         & $0.200$ & \ 45 & 6.38 \\
6  & Lasith Malinga      & $0.199$ & 170 & 6.77 \\
7  & Glenn McGrath       & $0.192$ & \ 12 & 6.52 \\
8  & Bhuvneshwar Kumar   & $0.186$ & 226 & 7.44 \\
9  & Dale Steyn          & $0.180$ & \ 97 & 6.61 \\
10 & M.\ Muralitharan    & $0.161$ & \ 63 & 6.45 \\
11 & Lungi Ngidi         & $0.151$ & \ 42 & 8.16 \\
12 & Rashid Khan         & $0.147$ & 179 & 7.17 \\
\bottomrule
\end{tabular}
\end{subtable}\hfill
\begin{subtable}{0.55\linewidth}
\centering
\subcaption{By cumulative Bowl Impact $\OmegaBowl$}
\label{tab:iplbowl_imp}
\scriptsize
\setlength{\tabcolsep}{3pt}
\begin{tabular}{rlrrrrrr}
\toprule
Rank & Bowler & $\OmegaBowl$ & $-\RAE$ & $\realDAR$ & $\xDAR$ & Econ & Wkts \\
\midrule
1  & Jasprit Bumrah      & $924$ & $1049$ & $750$ & $\ 875$ & 7.07 & 187 \\
2  & Sunil Narine        & $802$ & $1059$ & $728$ & $\ 985$ & 6.66 & 207 \\
3  & Bhuvneshwar Kumar   & $621$ & $\ 858$& $878$ & $1115$  & 7.44 & 226 \\
4  & Lasith Malinga      & $476$ & $\ 566$& $467$ & $\ 556$ & 6.77 & 170 \\
5  & Rashid Khan         & $430$ & $\ 520$& $701$ & $\ 790$ & 7.17 & 179 \\
6  & Dale Steyn          & $307$ & $\ 395$& $296$ & $\ 384$ & 6.61 & \ 97 \\
7  & Jofra Archer        & $296$ & $\ 345$& $371$ & $\ 419$ & 7.78 & \ 84 \\
8  & R.\ Ashwin          & $242$ & $\ 521$& $647$ & $\ 925$ & 6.95 & 187 \\
9  & M.\ Muralitharan    & $197$ & $\ 246$& $155$ & $\ 204$ & 6.45 & \ 63 \\
10 & Varun Chakravarthy  & $176$ & $\ 180$& $454$ & $\ 458$ & 7.39 & 111 \\
11 & Matheesha Pathirana & $175$ & $\ 208$& $154$ & $\ 188$ & 7.73 & \ 47 \\
12 & Anil Kumble         & $169$ & $\ 194$& $106$ & $\ 132$ & 6.38 & \ 45 \\
\bottomrule
\end{tabular}
\end{subtable}
\end{table}

Reading Table~\ref{tab:ipl_bat}(b) through its decomposition separates two distinct routes to a
high Bat Impact. The first is contextual run production. Virender Sehwag reaches $523$ almost
entirely through $\RAE$ ($598$, the second largest in the table), and is the only batter in the table whose wicket ledger is a net charge: his realised effect ($\realDAR=509$) exceeds its
expectation ($\xDAR=434$), so being an attacking opener, he was naturally dismissed more, and in higher-value states, than the league baseline predicts, a tax his scoring more than repays. That a batter who
last appeared in the IPL in $2015$ still ranks this high follows directly from the era factor
$\gamma^{\mathrm{era}}$: his stroke-play is scored against the lower bar of his own pre-inflation
seasons rather than the high-scoring modern game, so he is neither flattered by a friendly era nor
penalised for having missed the recent run-glut. The second route is survival. Virat Kohli
($\RAE=11$) and MS Dhoni ($\RAE=68$) add almost nothing above par per delivery, yet rank
equal-seventh at $601$, because their wicket ledger is a large credit: the expected effect far
exceeds the realised one ($\xDAR=1604$ against $\realDAR=1013$ for Kohli, $1069$ against $535$ for
Dhoni). Both figures are large because they scale with deliveries faced, and Kohli has faced more
than any batter in IPL history; their difference $\xDAR-\realDAR$ is the run value of the dismissals
each avoided relative to an average batter over that volume, which a raw wicket count cannot
see. KL Rahul ($672$, credit $467$) is the same anchor profile. AB de Villiers tops the
board by drawing on both routes at once, a large $\RAE$ ($713$) alongside a survival credit of
$310$.

The bowling decomposition reads differently from a raw wicket count. Jasprit Bumrah and Sunil
Narine lead on elite containment, the two largest run-saved totals in the league
($-\RAE$ around $1050$). Their realised wicket hauls sit close to, but do not exceed, what their heavy
volumes predict ($\realDAR=750$ against $\xDAR=875$ for Bumrah), so it is run prevention, not
wicket-taking, that elevates them. One thing the tables do not show directly: Yuzvendra Chahal took the most wickets of any bowler in IPL ($233$, $\realDAR=840$), yet his
expected haul over the same overs is almost identical ($\xDAR=845$), so he nets nothing significant
through wickets and ranks on his modest containment alone. R Ashwin shows the sharper edge of the
same mechanism: his realised wickets ($647$) fall well short of what his large volume would
ordinarily yield ($925$), so his standing rests on containment rather than strike. The pattern
holds across the table: once the expected haul is removed, the net wicket ledger is a modest,
usually negative adjustment, and it is run prevention that separates the leading bowlers, with the
wicket channel penalising those whose strike output lags their workload. Because Impact is a
cumulative runs total, longevity still contributes, but only through value added above
expectation, not through the sheer number of balls bowled.

As a sharper face-validity test, I rank single
innings by their batting Impact. Table~\ref{tab:top_innings} lists the top knocks of IPL history sorted by their Batting Impacts. It recovers the canonical great
innings, i.e.\ Chris Gayle's $175^\ast$ (the highest score in T20 history till date) on top, followed by Brendon McCullum's $158^\ast$ in the league opening match, and spans eras from $2008$ to $2025$ because $\RAE$
is measured against the scoring conditions of each era.

\begin{table}[!ht]
\centering
\caption{Most impactful individual IPL innings by batting Impact (innings
$\RAE-\DAR^{\text{bat}}$, in runs). Cost is the realised wicket value $\realDAR$ at the
dismissal, $0$ for an unbeaten innings; Credit is the expected wicket value $\xDAR$ accrued over
the deliveries faced, the survival dividend.}
\label{tab:top_innings}
\small
\begin{tabular}{rlrrrrrrrr}
\toprule
Rank & Batter & Opposition & Date & Runs & Balls & $\RAE$ & Cost & Credit & Impact \\
\midrule
1  & Chris Gayle      & PWI & 2013-04-23 & 175 & 66 & 76 & 0 & 15 & $91$ \\
2  & Brendon McCullum & RCB & 2008-04-18           & 158 & 73 & 60 & 0 & 17 & $77$ \\
3  & Virender Sehwag  & DC & 2011-05-05        & 119 & 56 & 55 & 4 & 15 & $66$ \\
4  & AB de Villiers   & GL & 2016-05-14       & 129 & 52 & 53 & 0 & 11 & $64$ \\
5  & David Miller     & RCB & 2013-05-06           & 101 & 38 & 54 & 0 & \ 9 & $63$ \\
6  & Yusuf Pathan     & MI & 2010-03-13        & 100 & 37 & 53 & 2 & 11 & $62$ \\
7  & David Warner     & KKR & 2017-04-30       & 126 & 59 & 52 & 3 & 13 & $62$ \\
8  & Priyansh Arya    & CSK & 2025-04-08       & 103 & 42 & 52 & 6 & 14 & $60$ \\
9  & AB de Villiers   & MI & 2015-05-10        & 133 & 59 & 46 & 0 & 13 & $59$ \\
10 & Andrew Symonds   & RR & 2008-04-24     & 117 & 53 & 45 & 0 & 13 & $58$ \\
\bottomrule
\end{tabular}
\end{table}

The innings ranking rewards contextual run value and survival rather than milestones, and this is
clearest among the knocks Table~\ref{tab:top_innings} leaves out. The most impactful sub-hundred
score in IPL history is Suresh Raina's $87$ from $25$ balls against Kings XI Punjab in $2014$ (Impact
$56$), almost at a level with a few lower hundreds in and around the table, despite Raina facing barely a quarter of the deliveries. Hardik Pandya's $91\,(34)$ and Rashid Khan's $79^\ast\,(32)$ reach the same Impact of $56$. A twenty-five-ball assault clears a far higher contextual bar than a measured hundred and the metric
records it as such. The costliest dismissals tell a different story: they are not the failures of
tail-enders but the early falls of set or top-order batters in the powerplay, where the wicket
value peaks. Travis Head's $22\,(12)$ cost $14$ runs, the largest single-dismissal debit in the
sample, and even the rare early failures of great players surface here, with Chris Gayle's
$10\,(10)$ and a sixteen-ball $16$ from MS Dhoni each costing $12$. The survival dividend, finally,
is largest for long unbeaten innings: Shikhar Dhawan's $99^\ast\,(66)$ earns a credit of $34$
runs and Dhoni's $63^\ast\,(45)$ earns a credit of $28$, where the batters were repaid for carrying their wicket through deliveries an average batter would not have survived.

Aggregating Impact over an entire IPL season rather than a
single innings gives a longer-run view of contribution. Tables~\ref{tab:ipl_batseason}
and~\ref{tab:ipl_bowlseason} list the most impactful batting and bowling seasons in IPL history,
each decomposed into its run and wicket components.

\begin{table}[H]
\centering
\caption{Most impactful individual IPL batting seasons, by season Bat Impact
$\OmegaBat=\RAE-(\realDAR-\xDAR)$ (minimum $150$ balls faced).}
\label{tab:ipl_batseason}
\footnotesize
\setlength{\tabcolsep}{4.5pt}
\begin{tabular}{rllrrrrrrrr}
\toprule
\# & Batter & Team & Season & Runs & SR & $\RAE$ & $\RAE$/B & $\realDAR$ & $\xDAR$ & $\OmegaBat$ \\
\midrule
1  & Vaibhav Suryavanshi & RR   & 2026 & 776 & 237.3 & 254 & $0.78$ & 73 & \ 79 & 260 \\
2  & Chris Gayle         & RCB  & 2011 & 608 & 183.1 & 192 & $0.58$ & 38 & \ 69 & 223 \\
3  & Andre Russell       & KKR  & 2019 & 511 & 204.4 & 170 & $0.68$ & 32 & \ 82 & 220 \\
4  & Chris Gayle         & RCB  & 2012 & 733 & 160.7 & 158 & $0.35$ & 49 & 109 & 218 \\
5  & David Warner        & SRH  & 2016 & 848 & 151.4 & 147 & $0.26$ & 70 & 130 & 208 \\
6  & Virat Kohli         & RCB  & 2016 & 973 & 152.0 & 110 & $0.17$ & 46 & 141 & 205 \\
7  & AB de Villiers      & RCB  & 2016 & 687 & 168.8 & 155 & $0.38$ & 46 & \ 95 & 204 \\
8  & Glenn Maxwell       & KXIP & 2014 & 552 & 187.8 & 194 & $0.66$ & 63 & \ 65 & 196 \\
9  & Rishabh Pant        & DC   & 2018 & 684 & 173.6 & 157 & $0.40$ & 54 & \ 89 & 192 \\
10 & Chris Gayle         & RCB  & 2013 & 720 & 156.9 & 112 & $0.24$ & 57 & 107 & 163 \\
\bottomrule
\end{tabular}
\end{table}

\begin{table}[!ht]
\centering
\caption{Most impactful individual IPL bowling seasons, by season Bowl Impact
$\OmegaBowl=-\RAE+(\realDAR-\xDAR)$ (minimum $120$ balls bowled).}
\label{tab:ipl_bowlseason}
\footnotesize
\setlength{\tabcolsep}{4.5pt}
\begin{tabular}{rllrrrrrrrr}
\toprule
\# & Bowler & Team & Season & Wkts & Econ & $-\RAE$ & $-\RAE$/B & $\realDAR$ & $\xDAR$ & $\OmegaBowl$ \\
\midrule
1  & Jasprit Bumrah    & MI  & 2024 & 20 & 6.08 & 159 & $0.51$ & \ 86 & \ 83 & 162 \\
2  & Sunil Narine      & KKR & 2013 & 25 & 5.42 & 152 & $0.40$ & \ 61 & \ 66 & 148 \\
3  & Jasprit Bumrah    & MI  & 2025 & 19 & 6.57 & 141 & $0.50$ & \ 76 & \ 71 & 146 \\
4  & Sunil Narine      & KKR & 2012 & 26 & 5.37 & 130 & $0.37$ & \ 67 & \ 53 & 144 \\
5  & Jasprit Bumrah    & MI  & 2019 & 20 & 6.48 & 158 & $0.42$ & \ 66 & \ 85 & 139 \\
6  & Bhuvneshwar Kumar & RCB & 2026 & 29 & 7.73 & 106 & $0.28$ & 151 & 119 & 138 \\
7  & Rashid Khan       & SRH & 2020 & 20 & 5.27 & 138 & $0.36$ & \ 73 & \ 79 & 132 \\
8  & Mohammed Siraj    & RCB & 2023 & 19 & 6.81 & 114 & $0.38$ & \ 88 & \ 80 & 123 \\
9  & Sunil Narine      & KKR & 2026 & 16 & 6.60 & 138 & $0.45$ & \ 61 & \ 76 & 123 \\
10 & Lasith Malinga    & MI  & 2011 & 29 & 5.61 & 107 & $0.28$ & \ 83 & \ 71 & 119 \\
\bottomrule
\end{tabular}
\end{table}

Season Impact rewards contextual efficiency over raw volume. The
most impactful batting season in IPL history belongs not to the leading run-scorer but to Vaibhav
Suryavanshi's $776$ runs at a strike rate of $237$ in $2026$ (Impact $260$), whose per-ball value
of $0.78$ is much higher than that of Virat Kohli's record $973$-run season in $2016$ (Impact $205$, per-ball
value $0.17$), which ranks only sixth. Several of the most impactful seasons rest on modest
aggregate returns, because what those runs were worth, ball for ball and situation for
situation, far outstripped a longer but flatter accumulation, for example, Andre Russell's $511$ runs in $2019$ (Impact $220$) and Glenn Maxwell's $552$
in $2014$ (Impact $196$). A few names recur, a sign of sustained rather than one-off impact. Chris Gayle holds three of the $10$ most impactful batting seasons
($2011$-$2013$) and Jasprit Bumrah holds three of the top $10$ bowling seasons, with Sunil Narine close
behind. The two routes to bowling value reappear at the season scale as well. Bumrah leads through
elite containment at an economy near $6$, whereas Bhuvneshwar Kumar's $2026$ ($29$ wickets) is
carried by a large wicket haul despite a higher economy. Finally, the same players surface on both
leaderboards: Narine ranks among the most impactful bowling seasons and, as an opener, among the
most valuable batting rates, and Pat Cummins appears on both sides too, the mark of season-long
all-round contribution.

\section{Concluding remarks}
\label{sec:limits}

The framework reduces T20 player evaluation to two ball-level primitives, both measured as runs
above expectation, and its strengths follow from that single design choice. $\RAE$ prices the runs
on every delivery against a contextual, likelihood-grounded expectation whose per-cell update is an
exact conditional Poisson maximum-likelihood multiplier (Proposition~\ref{prop:poisson}). The
stratified, nonparametric baseline makes it robust to playing conditions, phase and opposition
without assuming any distribution for the lumpy, over-dispersed run outcome, and a single
opposition symmetry lets the same residual serve both batting and bowling. Dismissal Adjusted Runs prices the
runs at every dismissal through a batting-side value function and then nets a player's realised
total against its hazard compensator, so the quantity that enters evaluation is a run-weighted
martingale residual, mean-zero by construction (Proposition~\ref{prop:centre}). This centering is
what turns a shared unit into a shared scale. Both primitives become above-expectation residuals,
so Bat Impact $=\sum\RAE-\DAR^{\text{bat}}$ and Bowl Impact $=-\sum\RAE+\DAR^{\text{bowl}}$ are
comparable in centre as well as in dimension (Finding~\ref{find:centering}), and a batter is
charged for a dismissal only to the extent it exceeds what his match situations warranted.

Some further properties follow from the same construction. Empirical-Bayes shrinkage, applied both to the
factor multipliers of \eqref{eq:rae_shrink} and to the player rates feeding the value function,
corrects the small-sample bias that inflates naive averages, and the symmetric opposition factors
remove any need for hand-set phase weights, since difficulty is carried inside the expectation
itself (Remark~\ref{rem:nophaseweight}). Because Dismissal Adjusted Runs is priced in the same unit as
$\RAE$, a bowler's two separate sources of value (run prevention and wicket-taking),
near-orthogonal by Finding~\ref{find:two_sources}, combine additively with no weighting constant to
defend, and every wicket is valued on its own match state rather than by a flat exchange rate
(Finding~\ref{find:dar_state}). Nothing is tuned in the construction: the weight of a wicket relative
to a run is a consequence of the value function and the hazard baseline, not a parameter I set.

The framework's limitations point directly to its extensions. The expectation model treats delivery
context as exogenous to player identity, whereas in practice bowlers adjust line and length to
specific batters; a batter who forces bowlers into uncomfortable areas may carry a slightly deflated
$\hat\mu(\bm{x}_b)$, and a structural model of bowler strategy would address the resulting
confounding. The second-innings value function conditions on the innings alone, capturing average
chase behaviour through innings-specific transition models but not the live required run rate; a
finer chase state, valuing runs and wickets against the current target in the spirit of the Pressure
Index of \citet{bandyopadhyay2026applications}, is a natural refinement. The value function itself tracks batter
quality by discrete tiers and treats every dismissal as the striker's, so run-outs of the
non-striker, finer within-tier heterogeneity and a full ball-by-ball settling model are left aside,
though these are secondary to the state variables that dominate a wicket's worth. On the input side,
$\RAE$ uses only always-available fields so that every match is scored identically; the line, length
and stroke-geometry data present for a minority of deliveries could sharpen the expected runs model
on the annotated subset, at the cost of comparability across the dataset. This extended $\RAE$ model, together with the evidence for adopting line and length as the additional covariates and for confining the extension to batting, is developed in Section~\ref{sec:supp:mrae} in the supplementary materials.

Two choices shape how the numbers should be read. The
expected $\DAR$ uses a league dismissal hazard $\bar h_{\varphi,e}$ conditioned on phase and
innings only; the coarseness is deliberate, since a context-level baseline lets batting skill
surface as survival above the league rate rather than being absorbed into a player-specific
expectation, but it leaves the survival credit indifferent to finer state such as the wicket count
or the tier of the batter at the crease. Conditioning $\bar h$ on the richer state the value
function already tracks would sharpen the compensator, at the cost of moving the baseline closer to
the player and diluting the skill signal it is meant to expose; where to set that resolution is
a modelling choice I leave open. A related consequence is that cumulative Impact is longevity-aware
by construction. Both the run value and the survival credit accumulate with volume, so a durable
accumulator can top the cumulative board on sustained competence while the per-ball rate isolates
peak quality, and the two are meant to be read together. Finally, because both primitives are
defined relative to the average player in the chosen data slice, cross-competition comparisons
inherit a system-wide competition-strength discrepancy, and a standard-adjusted baseline is planned.


\section*{Data availability statement}
The data used in this research have been obtained from three sources: baseline ball-by-ball data from the Cricsheet website (link: \href{https://cricsheet.org/matches/}{https://cricsheet.org/matches/}), granular ball-by-ball data from Dr. Himanish Ganjoo's collection of cricket data (link: \href{https://himanishganjoo.com/cricket-data/}{https://himanishganjoo.com/cricket-data/}) with his authorization and assistance, and hawk-eye ball-tracking data for IPL 2026 from the IPL website (link: \href{https://www.iplt20.com/matches/results}{https://www.iplt20.com/matches/results}). The cleaned data, along with all relevant Python code and other files, will be made publicly available in the author's GitHub repository.

\bibliographystyle{apalike}
\bibliography{references}

\clearpage
\setcounter{section}{0}\setcounter{subsection}{0}
\setcounter{table}{0}\setcounter{figure}{0}\setcounter{equation}{0}
\renewcommand{\thesection}{S\arabic{section}}
\renewcommand{\thesubsection}{S\arabic{section}.\arabic{subsection}}
\renewcommand{\thetable}{S\arabic{table}}
\renewcommand{\thefigure}{S\arabic{figure}}
\renewcommand{\theequation}{S\arabic{equation}}

\begin{center}
{\LARGE\bfseries Supplementary materials}
\end{center}
\vspace{1.6em}

\section{Model specification for reproducibility}
\label{sec:supp:spec}

This section records the fitted values of every quantity the estimators of
Sections~\ref{sec:rae}--\ref{sec:dar} produce, so that the numbers reported in the paper can be
reproduced exactly. Unless stated otherwise, each quantity is estimated on the full dataset of
$2{,}727{,}420$ legal deliveries of Section~\ref{sec:data}. The estimators themselves, i.e.\ the shrunk
backfitting for $\RAE$ and the tiering, transition model, and backward recursion for Dismissal Adjusted
Runs, are given in the main text and are not repeated. The three high-cardinality $\RAE$
factors (venue and the two opposition identities, with $1{,}268$, $7{,}282$, and $9{,}740$ cells)
are too numerous to list; each is estimated by the shrinkage rule~\eqref{eq:rae_shrink} and its
fitted extrema appear in Table~\ref{tab:rae_factors}.

\subsection{Global constants and fixed hyperparameters}
\label{sec:supp:const}

Every tuning quantity in the framework is either estimated from the data or fixed at a value stated
here; there are no hand-set weights. Table~\ref{tab:supp:const} collects the fixed constants. The
per-factor $\RAE$ shrinkage strengths $\kappa_f$ are themselves data-estimated by empirical Bayes
and are reported with the factor model in Table~\ref{tab:rae_factors}.

\begin{table}[H]
\centering
\caption{Global constants and fixed hyperparameters.}
\label{tab:supp:const}
\small
\begin{tabular}{lll}
\toprule
Symbol & Value & Meaning \\
\midrule
$\mu_0$              & $1.2425$        & grand mean runs per legal ball \\
$\Var(r)$            & $2.578$         & run variance (overdispersion scale) \\
$Q$                  & $6$             & batter quality tiers \\
$M_0$                & $120$           & legal balls in a T20 innings \\
phases               & $3$             & powerplay ($\leq$ over 6), middle (7--15), death (16--20) \\
$E_{\max}$           & $6$             & cap on the freshness horizon in \eqref{eq:delta_fresh} \\
venue window         & $3$ years       & bucket width for the venue factor \\
$\kappa_\mu$         & $60$            & EB pseudo-count for the batter scoring rate $\mu_i$ \eqref{eq:dar_rates} \\
$\kappa_\eta$        & $40$            & EB pseudo-count for the hazard $\eta_i$ and for $h_{i,\varphi,e}$ \\
$\kappa_\pi$         & $200$           & Dirichlet pseudo-count pooling $\pi_{r\mid i,\varphi,e}$ toward its phase mean \\
strike-share clip    & $[0.25,\,0.75]$ & bounds on the crease-window strike share $s_i$ \\
$\kappa_f$           & data-estimated  & per-factor $\RAE$ shrinkage (Table~\ref{tab:rae_factors}) \\
\bottomrule
\end{tabular}
\end{table}

\subsection{Expected-runs factor estimates}
\label{sec:supp:raefac}

The scenario and bowler-type multipliers appear in Table~\ref{tab:rae_cells} and the wicket factor
in Table~\ref{tab:rae_wkt}. Two further pieces complete the specification of the interpretable
factors. The bowler-type factor is resolved by phase in Table~\ref{tab:supp:btype} (Table~\ref{tab:rae_cells}
reports its phase average). The era factor, plotted in Figure~\ref{fig:rae_factors}, is listed in
full by year and phase in Table~\ref{tab:supp:era}.

\begin{table}[H]
\centering
\caption{Bowler-type multiplier $\gamma^{\mathrm{type}}$ by phase and bowling sub-type.}
\label{tab:supp:btype}
\small
\begin{tabular}{lccc}
\toprule
Sub-type & Powerplay & Middle & Death \\
\midrule
Right-arm pace  & 1.016 & 1.025 & 1.033 \\
Left-arm pace   & 0.951 & 1.026 & 1.007 \\
Leg-spin        & 1.056 & 0.998 & 0.950 \\
Off-spin        & 1.005 & 0.988 & 0.961 \\
Left-arm spin   & 0.988 & 0.951 & 0.953 \\
Other           & 0.990 & 1.001 & 1.009 \\
Unknown         & 0.958 & 0.985 & 1.011 \\
\bottomrule
\end{tabular}
\end{table}

\begin{table}[H]
\centering
\caption{Era multiplier $\gamma^{\mathrm{era}}$ by year and phase, with the phase means.}
\label{tab:supp:era}
\small
\begin{tabular}{lcccc}
\toprule
Year & Powerplay & Middle & Death & Mean \\
\midrule
2005 & 1.015 & 1.007 & 1.014 & 1.012 \\
2006 & 0.994 & 0.985 & 0.995 & 0.991 \\
2007 & 1.016 & 1.050 & 1.020 & 1.029 \\
2008 & 0.988 & 0.991 & 1.005 & 0.995 \\
2009 & 1.027 & 0.962 & 1.003 & 0.997 \\
2010 & 0.995 & 1.002 & 1.030 & 1.009 \\
2011 & 0.981 & 0.993 & 0.972 & 0.982 \\
2012 & 0.992 & 1.005 & 1.007 & 1.001 \\
2013 & 0.981 & 0.994 & 1.030 & 1.002 \\
2014 & 1.023 & 1.020 & 1.022 & 1.022 \\
2015 & 0.969 & 0.983 & 1.008 & 0.987 \\
2016 & 0.963 & 0.991 & 0.995 & 0.983 \\
2017 & 1.039 & 1.004 & 0.997 & 1.013 \\
2018 & 1.034 & 1.027 & 1.003 & 1.021 \\
2019 & 0.995 & 1.006 & 0.993 & 0.998 \\
2020 & 0.984 & 1.005 & 1.012 & 1.000 \\
2021 & 0.979 & 0.993 & 1.014 & 0.995 \\
2022 & 0.969 & 0.998 & 1.004 & 0.990 \\
2023 & 1.004 & 1.026 & 1.021 & 1.017 \\
2024 & 1.016 & 1.019 & 1.023 & 1.019 \\
2025 & 1.060 & 1.066 & 1.042 & 1.056 \\
2026 & 1.137 & 1.068 & 1.028 & 1.078 \\
\bottomrule
\end{tabular}
\end{table}

\subsection{Value-function transition model}
\label{sec:supp:dar}

The batting-side value function $\Vfun_e$ of Section~\ref{sec:dar:value} is fully determined by the
quality tiers (Table~\ref{tab:supp:tiers}), the per-state dismissal hazard $h_{i,\varphi,e}$
(Table~\ref{tab:supp:hazard}) and scoring distribution $\pi_{r\mid i,\varphi,e}$
(Tables~\ref{tab:supp:pi1}--\ref{tab:supp:pi2}), the incoming-batter tier $\iota(w)$ (the observed
average batting order of Section~\ref{sec:dar:trans}), and the set curve $f(k)$
(Table~\ref{tab:supp:setcurve}). Given these inputs, a single backward sweep of the
recursion~\eqref{eq:dar_bellman} reproduces $\Vfun_e$ exactly.

\begin{table}[H]
\centering
\caption{Quality tiers. A batter is assigned to a tier by his empirical Bayes shrunk scoring rate
$\mu_i$ falling in the stated range; the tier's ball-weighted mean rate $\bar\mu$ and mean hazard
$\bar\eta$ index the transition tables below.}
\label{tab:supp:tiers}
\small
\begin{tabular}{lcccc}
\toprule
Tier & $\mu_i$ range (runs/ball) & $\bar\mu$ & $\bar\eta$ & \#batters \\
\midrule
0 & $\mu_i < 1.016$        & 0.920 & 0.0268 & 1{,}373 \\
1 & $1.016 \le \mu_i < 1.103$ & 1.062 & 0.0273 & 1{,}669 \\
2 & $1.103 \le \mu_i < 1.171$ & 1.140 & 0.0250 & 2{,}011 \\
3 & $1.171 \le \mu_i < 1.241$ & 1.208 & 0.0230 & 2{,}402 \\
4 & $1.241 \le \mu_i < 1.333$ & 1.285 & 0.0218 & 1{,}342 \\
5 & $\mu_i \ge 1.333$         & 1.419 & 0.0244 & \ 943 \\
\bottomrule
\end{tabular}
\end{table}

\begin{table}[H]
\centering
\caption{Per-state dismissal hazard rate $h_{i,\varphi,e}$ (probability of a dismissal on a legal ball) by
tier $i$, phase $\varphi$ and innings $e$.}
\label{tab:supp:hazard}
\small
\begin{tabular}{lcccccc}
\toprule
 & \multicolumn{3}{c}{Innings 1} & \multicolumn{3}{c}{Innings 2} \\
\cmidrule(r){2-4}\cmidrule(r){5-7}
Tier & PP & Mid & Death & PP & Mid & Death \\
\midrule
0 & 0.0527 & 0.0590 & 0.1105 & 0.0539 & 0.0625 & 0.1030 \\
1 & 0.0503 & 0.0535 & 0.1028 & 0.0523 & 0.0565 & 0.1011 \\
2 & 0.0473 & 0.0485 & 0.0978 & 0.0507 & 0.0538 & 0.0922 \\
3 & 0.0435 & 0.0448 & 0.0910 & 0.0458 & 0.0488 & 0.0875 \\
4 & 0.0421 & 0.0446 & 0.0861 & 0.0435 & 0.0462 & 0.0796 \\
5 & 0.0443 & 0.0461 & 0.0823 & 0.0481 & 0.0489 & 0.0812 \\
\bottomrule
\end{tabular}
\end{table}

\begin{table}[H]
\centering
\caption{Scoring distribution $\pi_{r\mid i,\varphi,e}$ in first innings, by tier and phase.}
\label{tab:supp:pi1}
\footnotesize
\begin{tabular}{llcccccc}
\toprule
Tier & Phase & $0$ & $1$ & $2$ & $3$ & $4$ & $6$ \\
\midrule
0 & PP    & 0.620 & 0.242 & 0.050 & 0.005 & 0.075 & 0.008 \\
0 & Mid   & 0.510 & 0.371 & 0.061 & 0.003 & 0.044 & 0.012 \\
0 & Death & 0.448 & 0.383 & 0.085 & 0.005 & 0.057 & 0.021 \\
\addlinespace
1 & PP    & 0.553 & 0.265 & 0.050 & 0.006 & 0.108 & 0.017 \\
1 & Mid   & 0.421 & 0.431 & 0.065 & 0.003 & 0.058 & 0.021 \\
1 & Death & 0.384 & 0.404 & 0.091 & 0.005 & 0.078 & 0.037 \\
\addlinespace
2 & PP    & 0.526 & 0.270 & 0.050 & 0.006 & 0.126 & 0.022 \\
2 & Mid   & 0.374 & 0.453 & 0.071 & 0.004 & 0.069 & 0.029 \\
2 & Death & 0.348 & 0.403 & 0.098 & 0.004 & 0.095 & 0.051 \\
\addlinespace
3 & PP    & 0.502 & 0.278 & 0.050 & 0.006 & 0.138 & 0.026 \\
3 & Mid   & 0.340 & 0.468 & 0.074 & 0.004 & 0.078 & 0.036 \\
3 & Death & 0.304 & 0.410 & 0.105 & 0.004 & 0.109 & 0.067 \\
\addlinespace
4 & PP    & 0.478 & 0.284 & 0.049 & 0.006 & 0.149 & 0.034 \\
4 & Mid   & 0.321 & 0.473 & 0.073 & 0.004 & 0.086 & 0.043 \\
4 & Death & 0.283 & 0.401 & 0.108 & 0.004 & 0.120 & 0.083 \\
\addlinespace
5 & PP    & 0.460 & 0.271 & 0.049 & 0.006 & 0.163 & 0.051 \\
5 & Mid   & 0.318 & 0.454 & 0.071 & 0.003 & 0.092 & 0.063 \\
5 & Death & 0.282 & 0.374 & 0.105 & 0.003 & 0.126 & 0.110 \\
\bottomrule
\end{tabular}
\end{table}

\begin{table}[H]
\centering
\caption{Scoring distribution $\pi_{r\mid i,\varphi,e}$ in second innings, by tier and phase.}
\label{tab:supp:pi2}
\footnotesize
\begin{tabular}{llcccccc}
\toprule
Tier & Phase & $0$ & $1$ & $2$ & $3$ & $4$ & $6$ \\
\midrule
0 & PP    & 0.622 & 0.235 & 0.046 & 0.005 & 0.081 & 0.011 \\
0 & Mid   & 0.516 & 0.363 & 0.060 & 0.002 & 0.045 & 0.014 \\
0 & Death & 0.501 & 0.354 & 0.066 & 0.003 & 0.055 & 0.020 \\
\addlinespace
1 & PP    & 0.544 & 0.270 & 0.049 & 0.005 & 0.114 & 0.019 \\
1 & Mid   & 0.429 & 0.420 & 0.064 & 0.004 & 0.061 & 0.022 \\
1 & Death & 0.419 & 0.392 & 0.075 & 0.004 & 0.076 & 0.034 \\
\addlinespace
2 & PP    & 0.519 & 0.276 & 0.048 & 0.005 & 0.127 & 0.025 \\
2 & Mid   & 0.392 & 0.441 & 0.065 & 0.003 & 0.069 & 0.030 \\
2 & Death & 0.380 & 0.391 & 0.084 & 0.005 & 0.094 & 0.045 \\
\addlinespace
3 & PP    & 0.493 & 0.282 & 0.049 & 0.006 & 0.140 & 0.030 \\
3 & Mid   & 0.359 & 0.454 & 0.071 & 0.003 & 0.078 & 0.036 \\
3 & Death & 0.338 & 0.405 & 0.092 & 0.003 & 0.104 & 0.058 \\
\addlinespace
4 & PP    & 0.472 & 0.283 & 0.048 & 0.005 & 0.152 & 0.039 \\
4 & Mid   & 0.334 & 0.461 & 0.072 & 0.004 & 0.088 & 0.042 \\
4 & Death & 0.312 & 0.395 & 0.098 & 0.003 & 0.119 & 0.072 \\
\addlinespace
5 & PP    & 0.456 & 0.263 & 0.049 & 0.005 & 0.168 & 0.058 \\
5 & Mid   & 0.328 & 0.440 & 0.070 & 0.003 & 0.094 & 0.065 \\
5 & Death & 0.309 & 0.373 & 0.092 & 0.003 & 0.122 & 0.101 \\
\bottomrule
\end{tabular}
\end{table}

\begin{table}[H]
\centering
\caption{The set curve $f(k)$: the fraction of the settled scoring rate realised over a batter's
first $k$ deliveries, used in the freshness term~\eqref{eq:delta_fresh}.}
\label{tab:supp:setcurve}
\small
\begin{tabular}{l ccccccccccc}
\toprule
$k$    & 1 & 2 & 3 & 4 & 5 & 6 & 8 & 10 & 15 & 20 & 30 \\
\midrule
$f(k)$ & 0.623 & 0.688 & 0.734 & 0.771 & 0.802 & 0.830 & 0.876 & 0.912 & 0.970 & 1.003 & 1.042 \\
\bottomrule
\end{tabular}
\end{table}

The centered Dismissal Adjusted Runs additionally uses the league dismissal hazard $\bar
h_{\varphi,e}$. Table~\ref{tab:supp:hbar} gives it, and also reports how the hazard varies across
the major T20 leagues. The assembly of Impact then follows Algorithm~\ref{alg:impact}.

\begin{table}[H]
\centering
\caption{Empirical dismissal hazard $\bar h_{\varphi,e}$ (probability of a dismissal on a legal
ball) by phase and innings, within each of ten major T20 leagues and pooled over the full dataset.
The pooled values (final block) are the baseline $\bar h_{\varphi,e}$ that centers Dismissal Adjusted Runs
(Section~\ref{sec:dar}).}
\label{tab:supp:hbar}
\small
\begin{tabular}{llccc}
\toprule
League & Innings & Powerplay & Middle & Death \\
\midrule
\multirow{2}{*}{Indian Premier League}          & 1 & 0.0386 & 0.0417 & 0.0844 \\
                              & 2 & 0.0415 & 0.0446 & 0.0820 \\
\midrule
\multirow{2}{*}{Big Bash League}          & 1 & 0.0414 & 0.0441 & 0.0904 \\
                              & 2 & 0.0451 & 0.0489 & 0.0871 \\
\midrule
\multirow{2}{*}{Caribbean Premier League}          & 1 & 0.0440 & 0.0444 & 0.0899 \\
                              & 2 & 0.0457 & 0.0497 & 0.0833 \\
\midrule
\multirow{2}{*}{Pakistan Super League}          & 1 & 0.0387 & 0.0465 & 0.0916 \\
                              & 2 & 0.0457 & 0.0487 & 0.0904 \\
\midrule
\multirow{2}{*}{SA20}         & 1 & 0.0474 & 0.0448 & 0.0837 \\
                              & 2 & 0.0466 & 0.0547 & 0.0927 \\
\midrule
\multirow{2}{*}{Major League Cricket}          & 1 & 0.0513 & 0.0429 & 0.0860 \\
                              & 2 & 0.0464 & 0.0562 & 0.0798 \\
\midrule
\multirow{2}{*}{Syed Mushtaq Ali Trophy}         & 1 & 0.0454 & 0.0446 & 0.0911 \\
                              & 2 & 0.0425 & 0.0450 & 0.0820 \\
\midrule
\multirow{2}{*}{Super Smash}  & 1 & 0.0446 & 0.0483 & 0.0884 \\
                              & 2 & 0.0519 & 0.0519 & 0.0868 \\
\midrule
\multirow{2}{*}{Lanka Premier League}          & 1 & 0.0478 & 0.0486 & 0.0941 \\
                              & 2 & 0.0528 & 0.0512 & 0.0845 \\
\midrule
\multirow{2}{*}{T20 Blast}    & 1 & 0.0445 & 0.0461 & 0.0886 \\
                              & 2 & 0.0470 & 0.0505 & 0.0886 \\
\midrule
\multirow{2}{*}{Full dataset} & 1 & 0.0448 & 0.0471 & 0.0905 \\
                              & 2 & 0.0475 & 0.0502 & 0.0876 \\
\bottomrule
\end{tabular}
\end{table}

\newcommand{\pRAE}{\mathrm{pRAE}}
\newcommand{\mRAE}{\mathrm{mRAE}}
\section{Extending RAE with ball-tracking covariates}
\label{sec:supp:mrae}

The $\RAE$ model of Section~\ref{sec:rae} conditions only on covariates recorded for every delivery in our dataset,
so that it applies identically to every match and keeps players comparable across eras and
conditions. That choice is deliberate, but it leaves information unused: a
boundary struck off a yorker is a finer piece of batting than the same boundary off a half-volley,
yet the baseline model, blind to the delivery's line and length, scores the two alike. Where
ball-tracking data exist, we can refine the expectation to the exact delivery faced, and so isolate a
batter's shot-production more sharply. This section develops that refinement, measures how much it moves the evaluation, and says why
we keep it out of the headline metric and apply it to batting only.

To keep the two metrics separate we name them distinctly. We write $\pRAE$ (primary $\RAE$)
for the metric of the main paper and $\mRAE$ (modified $\RAE$) for the extension. $\pRAE$ is the residual $r_b-\mu(\bm x_b)$ against the
expectation built from fields recorded for every delivery (Equation~\eqref{eq:rae_model}), while $\mRAE$ is the residual against an
expectation that additionally conditions on the delivery's line and length. The availability of these two data tiers is described in Section~\ref{sec:data}. Because the two use
different covariate sets and, since tracking is partial, are defined on different samples, the
distinct notation guards against conflating a player's $\pRAE$ (over all his deliveries) with his
$\mRAE$ (over his tracked deliveries only).

\subsection{Ball-tracking data: covariates and coverage}
\label{sec:supp:mrae:data}

The ball-tracking tier records two kinds of variable. Delivery variables describe the ball
the bowler sent down, before the batter played it: its length (Yorker, Full Toss, Half Volley, Good
Length, Back of Length, Short, Bouncer), its line (On Stumps, Outside Off, Down Leg and their wide
variants), and the bowler's variation. Response variables describe what the batter then did:
whether the shot was controlled, the stroke played, its elevation, and the area to which the ball
went. Only the delivery variables are candidates for an expected-runs model, for the reason given in
Section~\ref{sec:supp:mrae:choose}.

Coverage is partial and uneven across competitions (Table~\ref{tab:supp:cov}). Line and length
are present for about one-third of all legal deliveries in men's T20 matches and only a quarter of men's T20 internationals, but for a
majority of the well-resourced franchise leagues: $57\%$ of the IPL, $66\%$ of the Big Bash and
around three quarters of the CPL and SA20 in our dataset carry ball-tracking data, while some competitions are barely tracked. Coverage is
also patchy within a competition: the T20 Blast carries length for one-third of its deliveries but line
for almost none. This heterogeneity is exactly why line and length cannot enter the headline metric, at least until ball-tracking data become publicly available across the board.
A residual built on them would be defined on an unpredictable, competition-dependent subset and would
not be comparable across matches, so the extension below is a refinement available
wherever tracking exists, and it sits alongside $\pRAE$ rather than taking its place.

\begin{table}[H]
\centering
\caption{Ball-tracking coverage: percentage of legal deliveries carrying each covariate, overall,
for men's internationals, and for the major franchise leagues.}
\label{tab:supp:cov}
\small
\begin{tabular}{lrcccc}
\toprule
Competition & Legal balls & Length & Line & Control & Shot type \\
\midrule
Overall (all data)        & 2{,}727{,}420 & 33.9 & 29.5 & 39.7 & 16.5 \\
Men's internationals      & 1{,}213{,}059 & 25.2 & 24.2 & 25.4 & 12.4 \\
\midrule
Indian Premier League     & \ 285{,}761 & 57.1 & 57.0 & 57.1 & 32.3 \\
Big Bash League           & \ 147{,}955 & 65.5 & 65.5 & 65.4 & 49.6 \\
Caribbean Premier League  & \ \ 90{,}866 & 75.7 & 75.7 & 75.7 & 22.6 \\
Pakistan Super League     & \ \ 81{,}037 & 60.8 & 60.8 & 60.9 & \ 0.0 \\
SA20                      & \ \ 28{,}048 & 75.5 & 75.5 & 75.4 & \ 0.0 \\
T20 Blast                 & \ 324{,}920 & 33.7 & \ 0.4 & 82.1 & 33.2 \\
Syed Mushtaq Ali Trophy   & \ 154{,}361 & \ 6.4 & \ 6.4 & \ 6.4 & \ 0.0 \\
Super Smash               & \ \ 59{,}327 & 15.1 & 15.1 & 15.0 & \ 0.0 \\
\bottomrule
\end{tabular}
\end{table}

\subsection{Choosing the covariates}
\label{sec:supp:mrae:choose}

An expected runs model may condition only on what is known before the outcome of a ball, which rules out
the response variables at once. Conditioning expected runs on whether a shot was controlled, or on
where the ball was hit, would fold the outcome into its own expectation and explain away the
run-scoring the residual is meant to capture. The objection is not only a conceptual one.
Table~\ref{tab:supp:sig} fits each candidate as a single phase-conditioned factor on top of the
baseline expectation and reports its between-cell signal $\tau^2$, its data-estimated shrinkage
$\kappa$, and the fraction of residual sum of squares it removes ($\Delta$RSS). The response
variables carry by far the largest apparent signal. Shot area alone removes $20\%$ of the residual
variance, with a multiplier range from $0.02$ to $1.75$, because they are proxies for the
result: a ball hit to the boundary explains its four runs tautologically. Admitting them would be
circular, so they are excluded despite, indeed because of, their apparent significance.

\begin{table}[H]
\centering
\caption{Significance of candidate additional covariates, each fitted as a single phase-conditioned
factor on top of the baseline expectation. $\Delta$RSS is the percentage of residual sum of squares
the factor removes on its own coverage subset. Response variables are post-outcome and are excluded
from the model despite their large apparent signal (see text).}
\label{tab:supp:sig}
\small
\begin{tabular}{llrrcrrr}
\toprule
Covariate & Role & Coverage (\%) & Cells & Multiplier range & $\tau^2$ & $\kappa$ & $\Delta$RSS (\%) \\
\midrule
Length    & delivery & 33.9 & 21 & $0.51$--$1.44$ & 0.0455 & \ 37 & 2.61 \\
Line      & delivery & 29.5 & 21 & $0.93$--$1.18$ & 0.0038 & 437 & 0.20 \\
Variation & delivery & 16.5 & 39 & $0.53$--$1.08$ & 0.0061 & 272 & 0.32 \\
\midrule
Control   & response & 39.7 & 15 & $0.01$--$1.17$ & 0.1282 & \ 13 & \ 6.47 \\
Shot type & response & 16.5 & 71 & $0.06$--$1.58$ & 0.1646 & \ 10 & 10.24 \\
Shot area & response & 16.1 & 48 & $0.02$--$1.75$ & 0.3164 & \ \ 5 & 19.93 \\
\bottomrule
\end{tabular}
\end{table}

Among the admissible delivery variables, length is by far the most informative: it removes
$2.6\%$ of the residual variance with a multiplier ranging from $0.49$ on a yorker to $1.46$ on a
short ball. Line is comparatively weak ($0.2\%$, almost all of it in the wide categories), and the
bowler's variation weak too ($0.3\%$) and poorly covered. Adding line to length lifts the removed
variance from $2.63\%$ to $2.74\%$, and adding variation to $2.86\%$. The marginal gains beyond length
are small, and variation's low coverage would further shrink the tracked sample. We therefore adopt
`line and length' as the additional covariates to formulate $\mRAE$.

\subsection{The modified expected runs model}
\label{sec:supp:mrae:model}

The extension is a strict addition to the multiplicative model~\eqref{eq:rae_model}, not a new
mechanism. On the tracked subset we append two factors,
\begin{equation}
\label{eq:mrae}
\tilde\mu(\bm x_b)=\mu(\bm x_b)\,\gamma^{\mathrm{len}}_{(\varphi,\,\ell(b))}\,\gamma^{\mathrm{lin}}_{(\varphi,\,\lambda(b))},
\end{equation}
where $\ell(b)$ and $\lambda(b)$ are the length and line of delivery $b$, each factor is
phase-conditioned like the others, and $\gamma^{\mathrm{len}}$, $\gamma^{\mathrm{lin}}$ are estimated
by the same iterated, empirical Bayes shrunk backfitting of
Sections~\ref{sec:rae:fit}-\ref{sec:rae:shrink}, holding the baseline expectation $\mu(\bm x_b)$ as the offset. The modified and
primary residuals are then
\begin{equation}
\label{eq:mrae_res}
\mRAE_b=r_b-\tilde\mu(\bm x_b),\qquad\qquad \pRAE_b=r_b-\mu(\bm x_b).
\end{equation}
The fitted multipliers (Table~\ref{tab:supp:llmult}) line up with cricketing intuition. A yorker is worth about
half an average ball and a short ball about half again as much, while line matters only through width,
down the leg side or wide outside off. Because length dominates, $\mRAE$ is in effect a length
adjustment, raising the bar on deliveries that are easy to score off and lowering it on the hardest.

\begin{table}[H]
\centering
\caption{Fitted length and line multipliers (phase-averaged).}
\label{tab:supp:llmult}
\small
\begin{tabular}{lc@{\hskip 3em}lc}
\toprule
Length & $\gamma^{\mathrm{len}}$ & Line & $\gamma^{\mathrm{lin}}$ \\
\midrule
Yorker         & 0.490 & Off Stump        & 0.973 \\
Bouncer        & 0.607 & Leg Stump        & 0.982 \\
Good Length    & 0.880 & Outside Off      & 0.986 \\
Back of Length & 0.970 & On Stumps        & 0.988 \\
Half Volley    & 1.188 & Down Leg         & 1.013 \\
Full Toss      & 1.320 & Wide Down Leg    & 1.042 \\
Short          & 1.460 & Wide Outside Off & 1.118 \\
\bottomrule
\end{tabular}
\end{table}

Line and length occupy opposite roles for the two sides of a delivery. For the batter, they provide context: the ball he had to deal with and did not choose, so conditioning on them isolates his
shot-making given the delivery. For the bowler, they measure skill: bowling a yorker or hitting a
hard length is the act by which he suppresses scoring, so conditioning a bowler's residual on his
own line and length subtracts out his main contribution, just as conditioning a batter's
residual on his own shot selection would. The distinction mirrors the opposition factor asymmetry
already in the model (Section~\ref{sec:rae:agg}), where a batter is credited against the bowler he
faced and a bowler against the batter he faced, never against his own quality.

Applied to batters, the extension leaves the ranking essentially intact:
across the $107$ qualifying IPL batters ($\ge400$ tracked balls) the $\pRAE$ and $\mRAE$ per-ball
rates correlate $0.984$, with Spearman rank correlation $0.97$. Applied to bowlers, conditioning on
their own line and length degrades the ranking markedly (Spearman $0.83$) and moves it in the wrong
direction. The bowlers who lose most are the ones whose value is their length and line
control: Rashid Khan ($+0.170\to+0.118$ runs saved per ball), Bhuvneshwar Kumar ($+0.178\to+0.117$),
Sunil Narine ($+0.168\to+0.099$) and Mitchell Santner ($+0.139\to+0.078$), since the model now
treats the `hard to score' deliveries they bowl as the expected state of affairs rather than their
achievement. We therefore define $\mRAE$ for batters only.

\subsection{Primary RAE versus modified RAE}
\label{sec:supp:mrae:results}

Ball tracking is available in the IPL only from the $2015$ season onward, so every leaderboard in
this subsection is computed on the tracked deliveries of the $2015$--$2025$ seasons, and each
player's runs, balls, strike rate, wickets and economy are likewise measured over those tracked
deliveries alone, so that $\pRAE$ and $\mRAE$ are compared on exactly
the same balls. Each table is ordered by the modified metric, and we report two ranks: the
`mRank' is the player's position under $\mRAE$ and the `pRank' is the position he would have held under $\pRAE$ on the same balls. The gap between the two is the clearest single summary of what the tracking refinement changes.

\begin{table}[H]
\centering
\caption{Top IPL batters by $\mRAE$ per ball on the tracked subset ($2015$--$2025$,
$\ge400$ tracked balls), with $\pRAE$ per ball on the same balls, the two ranks, the change
$\Delta=\mRAE-\pRAE$, and runs, balls and strike rate over the tracked deliveries.}
\label{tab:supp:mrael}
\footnotesize
\setlength{\tabcolsep}{4.5pt}
\begin{tabular}{rlrrrrrrr}
\toprule
mRank & Batter & $\mRAE$/Ball & $\pRAE$/Ball & pRank & $\Delta$ & Runs & Balls & SR \\
\midrule
1  & Andre Russell    & $0.321$ & $0.375$ & 1  & $-0.054$ & 2{,}588 & 1{,}475 & 175.5 \\
2  & AB de Villiers   & $0.249$ & $0.270$ & 5  & $-0.021$ & 3{,}102 & 1{,}919 & 161.6 \\
3  & Nicholas Pooran  & $0.220$ & $0.272$ & 4  & $-0.052$ & 2{,}293 & 1{,}357 & 169.0 \\
4  & Phil Salt        & $0.215$ & $0.272$ & 3  & $-0.057$ & 1{,}056 & \ 600 & 176.0 \\
5  & Heinrich Klaasen & $0.204$ & $0.235$ & 7  & $-0.031$ & 1{,}480 & \ 872 & 169.7 \\
6  & Pat Cummins      & $0.192$ & $0.205$ & 10 & $-0.013$ & \ 612 & \ 402 & 152.2 \\
7  & Travis Head      & $0.186$ & $0.245$ & 6  & $-0.059$ & 1{,}146 & \ 674 & 170.0 \\
8  & Sunil Narine     & $0.179$ & $0.310$ & 2  & $-0.131$ & 1{,}733 & 1{,}028 & 168.6 \\
9  & Abhishek Sharma  & $0.176$ & $0.218$ & 8  & $-0.042$ & 1{,}816 & 1{,}114 & 163.0 \\
10 & Liam Livingstone & $0.165$ & $0.209$ & 9  & $-0.044$ & 1{,}051 & \ 662 & 158.8 \\
11 & Glenn Maxwell    & $0.161$ & $0.185$ & 12 & $-0.024$ & 2{,}223 & 1{,}482 & 150.0 \\
12 & Prithvi Shaw     & $0.146$ & $0.161$ & 14 & $-0.015$ & 1{,}884 & 1{,}280 & 147.2 \\
13 & Suryakumar Yadav & $0.132$ & $0.152$ & 16 & $-0.020$ & 4{,}140 & 2{,}776 & 149.1 \\
14 & Shashank Singh   & $0.128$ & $0.157$ & 15 & $-0.029$ & \ 773 & \ 490 & 157.8 \\
15 & Rajat Patidar    & $0.128$ & $0.129$ & 21 & $-0.001$ & 1{,}111 & \ 720 & 154.3 \\
\bottomrule
\end{tabular}
\end{table}

\begin{table}[H]
\centering
\caption{Top IPL bowlers by $-\mRAE$ per ball (runs saved) on the tracked subset ($2015$--$2025$,
$\ge300$ tracked balls), shown only to demonstrate the failure of the extension on bowling
(Section~\ref{sec:supp:mrae:model}).}
\label{tab:supp:mrael_bwl}
\footnotesize
\setlength{\tabcolsep}{4.5pt}
\begin{tabular}{rlrrrrrrr}
\toprule
mRank & Bowler & $-\mRAE$/Ball & $-\pRAE$/Ball & pRank & $\Delta$ & Wkts & Econ & Balls \\
\midrule
1  & Matheesha Pathirana & $0.311$ & $0.301$ & 1  & $+0.010$ & \ 49 & 7.66 & \ 682 \\
2  & Jasprit Bumrah      & $0.293$ & $0.298$ & 2  & $-0.005$ & 190 & 6.93 & 3{,}086 \\
3  & Jofra Archer        & $0.291$ & $0.257$ & 3  & $+0.034$ & \ 60 & 7.49 & 1{,}228 \\
4  & Chris Morris        & $0.160$ & $0.117$ & 11 & $+0.043$ & \ 86 & 7.69 & 1{,}420 \\
5  & Harshal Patel       & $0.150$ & $0.080$ & 19 & $+0.070$ & 155 & 8.37 & 2{,}129 \\
6  & Andrew Tye          & $0.137$ & $0.113$ & 12 & $+0.024$ & \ 48 & 8.25 & \ 684 \\
7  & Nathan Ellis        & $0.136$ & $0.142$ & 8  & $-0.006$ & \ 20 & 8.44 & \ 381 \\
8  & Prasidh Krishna     & $0.135$ & $0.102$ & 14 & $+0.033$ & \ 79 & 8.32 & 1{,}516 \\
9  & Mitchell Marsh      & $0.135$ & $0.079$ & 20 & $+0.056$ & \ 21 & 8.10 & \ 300 \\
10 & Nathan Coulter-Nile & $0.131$ & $0.079$ & 21 & $+0.052$ & \ 49 & 7.49 & \ 815 \\
11 & Lasith Malinga      & $0.130$ & $0.098$ & 15 & $+0.032$ & \ 53 & 8.08 & \ 904 \\
12 & Morne Morkel        & $0.121$ & $0.058$ & 28 & $+0.063$ & \ 22 & 7.70 & \ 384 \\
13 & Mohammed Siraj      & $0.120$ & $0.105$ & 13 & $+0.015$ & 112 & 8.25 & 2{,}307 \\
14 & Rashid Khan         & $0.119$ & $0.171$ & 6  & $-0.052$ & 166 & 6.96 & 3{,}156 \\
15 & Bhuvneshwar Kumar   & $0.117$ & $0.178$ & 4  & $-0.061$ & 165 & 7.71 & 3{,}260 \\
\bottomrule
\end{tabular}
\end{table}

\paragraph{Per-ball rates.} Among batters (Table~\ref{tab:supp:mrael}), the ordering holds broadly but
reshuffles underneath. AB de Villiers rises from pRank $5$ to second under $\mRAE$, while
Sunil Narine falls from pRank $2$ to mRank $8$, the sharpest mover, because much of his scoring
comes off short and wide balls the refined model now expects runs from. The bowler table
(Table~\ref{tab:supp:mrael_bwl}) is the evidence, as mentioned in Section~\ref{sec:supp:mrae:model},
that the same refinement must not be applied to bowling. Conditioning a bowler's residual on his own
line and length lifts almost every bowler's apparent run prevention. Harshal Patel jumps from pRank
$19$ to mRank $5$, and Chris Morris and Morne Morkel climb similarly, while the elite
suppressors fall. Rashid Khan slides from pRank $6$ to mRank $14$ and Bhuvneshwar Kumar from pRank
$4$ to mRank $15$, because the hard lengths that are their skill now count as par. The resulting order does not describe bowling quality.

\paragraph{Cumulative and season totals.} The cumulative and per-season orderings tell the same
story on the batting side, where they are meaningful. Over careers,
(Table~\ref{tab:supp:mrae_batcum}) de Villiers and Russell trade the top two places and Narine slides
from pRank $8$ to mRank $12$. By season, (Table~\ref{tab:supp:mrae_batseas}) Russell's $2019$ and de
Villiers's $2016$ remain the most valuable tracked campaigns under either metric. The bowling
cumulative and season tables (Tables~\ref{tab:supp:mrae_bwlcum} and~\ref{tab:supp:mrae_bwlseas}) are
included for completeness and show the same distortion as the rates: Bumrah, whose value is as much
containment as length, is the one elite name to hold the top; Bhuvneshwar and Rashid each shed
about $200$ runs of apparent value, and Bumrah alone accounts for four of the ten most
valuable bowling seasons.

\begin{table}[H]
\centering
\caption{Top IPL batters by cumulative $\mRAE$ (tracked deliveries, $2015$--$2025$, $\ge400$ balls),
with cumulative $\pRAE$ on the same balls.}
\label{tab:supp:mrae_batcum}
\footnotesize
\setlength{\tabcolsep}{4.5pt}
\begin{tabular}{rlrrrrrrr}
\toprule
mRank & Batter & $\mRAE$ & $\pRAE$ & pRank & $\Delta$ & Runs & Balls & SR \\
\midrule
1  & AB de Villiers   & 479 & 517 & 2  & $-38$  & 3{,}102 & 1{,}919 & 161.6 \\
2  & Andre Russell    & 474 & 554 & 1  & $-80$  & 2{,}588 & 1{,}475 & 175.5 \\
3  & Suryakumar Yadav & 367 & 422 & 3  & $-55$  & 4{,}140 & 2{,}776 & 149.1 \\
4  & Jos Buttler      & 311 & 331 & 7  & $-20$  & 4{,}119 & 2{,}756 & 149.5 \\
5  & Rishabh Pant     & 299 & 357 & 5  & $-58$  & 3{,}546 & 2{,}401 & 147.7 \\
6  & Nicholas Pooran  & 298 & 369 & 4  & $-71$  & 2{,}293 & 1{,}357 & 169.0 \\
7  & Sanju Samson     & 260 & 257 & 10 & $+3$   & 4{,}156 & 2{,}929 & 141.9 \\
8  & Glenn Maxwell    & 238 & 274 & 9  & $-36$  & 2{,}223 & 1{,}482 & 150.0 \\
9  & David Warner     & 238 & 332 & 6  & $-94$  & 4{,}597 & 3{,}243 & 141.8 \\
10 & Abhishek Sharma  & 196 & 243 & 11 & $-47$  & 1{,}816 & 1{,}114 & 163.0 \\
11 & Prithvi Shaw     & 187 & 207 & 12 & $-19$  & 1{,}884 & 1{,}280 & 147.2 \\
12 & Sunil Narine     & 185 & 318 & 8  & $-134$ & 1{,}733 & 1{,}028 & 168.6 \\
\bottomrule
\end{tabular}
\end{table}

\begin{table}[H]
\centering
\caption{Top IPL bowlers by cumulative $-\mRAE$ (runs saved; tracked deliveries, $2015$--$2025$,
$\ge300$ balls), with $-\pRAE$ on the same balls. Included for completeness only
(Section~\ref{sec:supp:mrae:model}). $\Delta=(-\mRAE)-(-\pRAE)$.}
\label{tab:supp:mrae_bwlcum}
\footnotesize
\setlength{\tabcolsep}{4.5pt}
\begin{tabular}{rlrrrrrrr}
\toprule
mRank & Bowler & $-\mRAE$ & $-\pRAE$ & pRank & $\Delta$ & Wkts & Econ & Balls \\
\midrule
1  & Jasprit Bumrah    & 903 & 920 & 1  & $-17$  & 190 & 6.93 & 3{,}086 \\
2  & Bhuvneshwar Kumar & 381 & 581 & 2  & $-200$ & 165 & 7.71 & 3{,}260 \\
3  & Rashid Khan       & 374 & 539 & 4  & $-165$ & 166 & 6.96 & 3{,}156 \\
4  & Jofra Archer      & 357 & 315 & 6  & $+42$  & \ 60 & 7.49 & 1{,}228 \\
5  & Harshal Patel     & 319 & 169 & 11 & $+150$ & 155 & 8.37 & 2{,}129 \\
6  & Sunil Narine      & 319 & 541 & 3  & $-222$ & 135 & 6.98 & 3{,}224 \\
7  & Sandeep Sharma    & 315 & 346 & 5  & $-31$  & 135 & 7.74 & 2{,}713 \\
8  & Mohammed Siraj    & 277 & 242 & 7  & $+35$  & 112 & 8.25 & 2{,}307 \\
9  & Chris Morris      & 227 & 166 & 12 & $+62$  & \ 86 & 7.69 & 1{,}420 \\
10 & Matheesha Pathirana & 212 & 206 & 9 & $+6$  & \ 49 & 7.66 & \ 682 \\
11 & Prasidh Krishna   & 205 & 154 & 13 & $+50$  & \ 79 & 8.32 & 1{,}516 \\
12 & Pat Cummins       & 147 & \ 77 & 23 & $+70$ & \ 80 & 8.51 & 1{,}595 \\
\bottomrule
\end{tabular}
\end{table}

Not every player falls under $\mRAE$ and the direction of the change
is itself informative. Among batters, only eleven of the $107$ qualifiers gain, and they are the
orthodox accumulators who score against good balls rather than by waiting for loose ones: Faf du
Plessis (pRank $26\to$ mRank $15$ on the cumulative list, $+39$ runs), Ruturaj Gaikwad, Shane Watson
and Ajinkya Rahane, so the refinement redistributes credit from those who feast on hittable
deliveries towards those who milk good ones. On the bowling side, the pattern reverses and is far
larger. A majority of qualifying bowlers ($73$ of $132$) \emph{rise}, because conditioning away line
and length compresses the field, flattering the expensive bowler who at least bowled testing areas
and docking the miser who earned his economy through control. That the adjustment moves most bowlers
in the wrong direction is the quantitative case for confining $\mRAE$ to batting.

\begin{table}[!ht]
\centering
\caption{Most valuable IPL batter-seasons by $\mRAE$ (tracked deliveries, $2015$--$2025$, $\ge150$
balls in the season), with $\pRAE$ on the same balls, the two ranks, the per-ball change
$\Delta/\text{Ball}$, and tracked runs, balls and strike rate.}
\label{tab:supp:mrae_batseas}
\footnotesize
\setlength{\tabcolsep}{4.5pt}
\begin{tabular}{rllrrrrrrr}
\toprule
mRank & Batter & Season & $\mRAE$ & $\pRAE$ & pRank & $\Delta/\text{B}$ & Runs & Balls & SR \\
\midrule
1  & Andre Russell   & 2019 & 154 & 170 & 1  & $-0.064$ & 511 & 250 & 204.4 \\
2  & AB de Villiers  & 2016 & 151 & 156 & 3  & $-0.011$ & 686 & 406 & 169.0 \\
3  & Rishabh Pant    & 2018 & 145 & 158 & 2  & $-0.032$ & 679 & 390 & 174.1 \\
4  & David Warner    & 2016 & 130 & 149 & 4  & $-0.033$ & 848 & 559 & 151.7 \\
5  & Abhishek Sharma & 2024 & 128 & 140 & 6  & $-0.053$ & 484 & 237 & 204.2 \\
6  & Travis Head     & 2024 & 127 & 141 & 5  & $-0.045$ & 567 & 296 & 191.6 \\
7  & Suryakumar Yadav& 2023 & 125 & 137 & 7  & $-0.034$ & 605 & 334 & 181.1 \\
8  & AB de Villiers  & 2018 & 103 & 108 & 9  & $-0.017$ & 478 & 273 & 175.1 \\
9  & Prithvi Shaw    & 2021 & 103 & \ 98 & 16 & $+0.016$ & 479 & 301 & 159.1 \\
10 & Priyansh Arya   & 2025 & 103 & 108 & 10 & $-0.017$ & 545 & 299 & 182.3 \\
\bottomrule
\end{tabular}
\end{table}

\begin{table}[H]
\centering
\caption{Most valuable IPL bowler-seasons by $-\mRAE$ (runs saved; tracked deliveries,
$2015$--$2025$, $\ge120$ balls in the season), with $-\pRAE$ on the same balls. Included for
completeness only (Section~\ref{sec:supp:mrae:model}).}
\label{tab:supp:mrae_bwlseas}
\footnotesize
\setlength{\tabcolsep}{4.5pt}
\begin{tabular}{rllrrrrrrr}
\toprule
mRank & Bowler & Season & $-\mRAE$ & $-\pRAE$ & pRank & $\Delta/\text{B}$ & Balls & Wkts & Econ \\
\midrule
1  & Jasprit Bumrah  & 2019 & 149 & 158 & 2  & $-0.023$ & 374 & 21 & 6.48 \\
2  & Jasprit Bumrah  & 2024 & 142 & 159 & 1  & $-0.056$ & 315 & 21 & 6.08 \\
3  & Jasprit Bumrah  & 2025 & 139 & 141 & 4  & $-0.007$ & 284 & 21 & 6.57 \\
4  & Lasith Malinga  & 2015 & 138 & 108 & 10 & $+0.081$ & 362 & 24 & 7.08 \\
5  & Jofra Archer    & 2020 & 133 & 115 & 7  & $+0.053$ & 336 & 20 & 6.29 \\
6  & Prasidh Krishna & 2025 & 124 & 102 & 11 & $+0.061$ & 354 & 26 & 8.08 \\
7  & Rashid Khan     & 2020 & 118 & 138 & 5  & $-0.053$ & 384 & 20 & 5.27 \\
8  & Mohammed Siraj  & 2023 & 116 & 114 & 8  & $+0.005$ & 302 & 20 & 6.81 \\
9  & Jasprit Bumrah  & 2018 & 115 & 117 & 6  & $-0.007$ & 325 & 17 & 6.55 \\
10 & Sunil Narine    & 2022 & 110 & 142 & 3  & $-0.095$ & 336 & 10 & 5.45 \\
\bottomrule
\end{tabular}
\end{table}

Two caveats temper even the batting refinement, and together they are why the main paper headlines
$\pRAE$. First, it is available only on the tracked subset, large for some leagues and negligible for
others (Table~\ref{tab:supp:cov}), so $\mRAE$ cannot serve as a universal, reproducible cross-competition scale.
Second, treating the delivery as exogenous ignores the batter's own hand in it: a fine player
forces bowlers into the loose lengths he then punishes, and $\mRAE$, by adjusting them away, declines
to credit him for it (the confounding noted in Section~\ref{sec:limits}). $\mRAE$ is thus best read as
a sharper, tracking-aware companion to $\pRAE$ on the matches that support it, not as its direct replacement.

\end{document}